\documentclass[journal ]{new-aiaa}
\usepackage[utf8]{inputenc}
\usepackage{textcomp}
\newtheorem{theorem}{Theorem}[section]
\newtheorem{lemma}[theorem]{Lemma}

\newtheorem{definition}[theorem]{Definition}

\newenvironment{proof}{\paragraph{Proof:}}{\hfill$\square$}

\usepackage{algorithm}
\usepackage[noend]{algpseudocode}
\usepackage{amsmath}
\usepackage{subcaption}
\usepackage{mathtools}
\usepackage{amscd}
\usepackage{upgreek}
\usepackage{amsfonts}
\usepackage{graphicx}
\usepackage{epstopdf}
\usepackage{caption}
\usepackage{setspace}
\usepackage{tikz}
\usepackage{color}
\usepackage{graphicx}
\usepackage{scalerel,stackengine}
\stackMath
\newcommand\reallywidehat[1]{%
	\savestack{\tmpbox}{\stretchto{%
			\scaleto{%
				\scalerel*[\widthof{\ensuremath{#1}}]{\kern-.6pt\bigwedge\kern-.6pt}%
				{\rule[-\textheight/2]{1ex}{\textheight}}
			}{\textheight}%
		}{0.5ex}}%
	\stackon[1pt]{#1}{\tmpbox}%
}
\usepackage{comment}
\usepackage{amsmath}
\usepackage[version=4]{mhchem}
\usepackage{siunitx}
\usepackage{longtable,tabularx}
\title{Model Predictive Static Programming for Discrete-Time Optimal Control on Lie Groups}

\author{Akhil B Krishna 
	 \footnote{Research Engineer, Center for Artificial Intelligence and Robotics (CAIR), New York University Abu Dhabi, bkakhil24@gmail.com} and
	Mangal Kothari 
	\footnote{Mangal Kothari was with the Department of Aerospace Engineering, Indian Institute of Technology Kanpur, Kanpur, India. He is currently with ADSAI, EDGE Group, Abu Dhabi, UAE. This work was performed while he was with the Indian Institute of Technology Kanpur.  Email: mangalgnc@gmail.com } }
\begin{document}
	
\maketitle

\begin{abstract}
This paper extends the Model Predictive Static Programming (MPSP) framework for nonlinear systems evolving on Euclidean spaces to the simple mechanical systems evolving on Lie groups. Classical optimal control approaches based on Pontryagin’s Maximum Principle (PMP) lead to nonlinear two-point boundary value problems (TPBVPs), whose numerical solution becomes particularly challenging on nonlinear configuration spaces. To overcome this difficulty, the proposed Lie-group MPSP framework reformulates the finite-horizon optimal control problem as a sequence of static quadratic programs that admit closed-form control updates, thereby avoiding the need to solve TPBVPs directly. The development relies on left-trivialized variations, intrinsic linearization on Lie groups, and a recursive computation of terminal sensitivity matrices, which together enable computationally efficient real-time implementation. The proposed method is demonstrated through optimal flipping maneuvers of a variable-pitch quadrotor (VPQ) and a single-main-rotor helicopter (SMRH), both of which are capable of generating negative thrust. For validation, continuous-time necessary and sufficient optimality conditions are derived, and the corresponding TPBVP solutions are compared against the trajectories generated by the proposed MPSP method in numerical simulations. In addition, the proposed algorithm is systematically compared with the iterative Linear Quadratic Regulator (iLQR) method, and a detailed numerical study is presented to highlight the relative performance and computational features of the two approaches.

\end{abstract}

\section{Introduction}\label{Sec:Introduction}

Optimal control of nonlinear dynamical systems evolving on Lie groups has attracted significant attention due to its direct relevance in aerospace, robotics, and mechanical systems. In many practical applications, such as rigid body attitude dynamics on $\mathrm{SO}(3)$ and coupled pose dynamics on $\mathrm{SE}$(3), the configuration space possesses a nonlinear geometric structure that cannot be accurately captured using Euclidean approximations. The Pontryagin maximum principle (PMP) provides first-order necessary conditions for optimal control problems on Euclidean spaces. However, the configuration space of many mechanical systems is not a Euclidean space but a curved space, typically modeled as a smooth manifold. For instance, nonlinear oscillators may evolve on higher-dimensional unit spheres. In such settings, conventional calculus on Euclidean spaces is not directly applicable. A smooth manifold is formed by a collection of overlapping open sets, called charts, each diffeomorphic to an open subset of a Euclidean space; thus, a smooth manifold is locally Euclidean. A simple mechanical system on a smooth manifold is characterized by the manifold as its configuration space, a smooth Lagrangian, and a set of external forces acting as control inputs. Motivated by this structure, several intrinsic controllers have been proposed in the literature for trajectory tracking of simple mechanical systems, without recourse to local coordinate charts \cite{FB-ADL:04,Bullo1999TrackingFF,maithripala2006almost,maithripala2015intrinsic,sanyal2009inertia,lee2006optimal,lee2012robust,lee2008time}. A Lie group is both a smooth manifold and a topological group equipped with smooth group operations. For example, the configuration space of a quadrotor is the special Euclidean group $\mathrm{SE}(3)$, while that of a rigid body is the special orthogonal group $\mathrm{SO}(3)$.

Several works have addressed optimal control on specific Lie groups. Discrete-time necessary and sufficient conditions for optimality have been derived for $\mathrm{SO}(3)$ in \cite{lee2008time,kulumani2017constrained} and for $\mathrm{SE}(3)$ in \cite{lee2006optimal}. More general formulations were developed in \cite{kobilarov2011discrete}, which extended optimality conditions to arbitrary matrix Lie groups \cite{joshi2021robust} and in \cite{phogat2018discrete}, which incorporated discrete mechanics. Subsequent works  further expanded these ideas and applied them to problems such as optimal attitude maneuvers of a swivelling biplane \cite{jirwankar2022discrete}. Despite these advances, PMP-based approaches invariably lead to two-point boundary value problems (TPBVPs). Solving such TPBVPs on nonlinear manifolds requires Lie‑group shooting methods, which are computationally intensive and sensitive to initialization, making them unsuitable for real-time applications.

In parallel, several real-time trajectory optimization methods have been developed for nonlinear systems. Differential Dynamic Programming (DDP) \cite{mayne1966ddp,jacobson1970ddp} and its simplified variant, iterative Linear Quadratic Regulation (iLQR) \cite{li2004ilqr,tassa2012synthesis}, are widely used in robotics and aerospace applications. These methods compute locally optimal control sequences through second-order expansions of the dynamics and backward Riccati recursions. While geometric variants of DDP have been proposed for systems evolving on Lie groups, they remain computationally demanding and sensitive to initialization, particularly for aggressive finite-time maneuvers. This motivates the search for alternative real-time optimal control methods that retain computational efficiency while respecting the underlying Lie-group structure.

Model Predictive Static Programming (MPSP), introduced in \cite{kothari2007hybrid,padhi2009model,kothari2010nonlinear}, offers an attractive alternative. MPSP combines ideas from nonlinear model predictive control and approximate dynamic programming, but avoids costate propagation and Riccati-type recursions. The key structural feature enabling this computational efficiency is that the MPSP cost penalizes only the control effort, while the terminal state is enforced as a hard constraint. This eliminates the need to propagate costates or solve Riccati equations, causing the optimality conditions to collapse into a static quadratic program with closed-form control updates. Consequently, MPSP intentionally trades exact optimality for real-time tractability. Unlike PMP-based methods, which require solving a nonlinear TPBVP, MPSP relies on local sensitivity information to generate fast closed-form control updates. Although many systems of interest, such as those evolving on $\mathrm{SO}(3)$ and $\mathrm{SE}(3)$, admit matrix representations, the formulation developed in this paper applies MPSP intrinsically on Lie groups and does not rely on any specific coordinate embedding.

MPSP has been successfully applied to a variety of aerospace guidance and control problems \cite{kothari2007hybrid,padhi2009model,kothari2010nonlinear}. However, existing formulations are confined to Euclidean spaces and do not incorporate the intrinsic geometry of nonlinear configuration spaces such as $\mathrm{SO}(3)$ and $\mathrm{SE}(3)$. Beyond its role as a real-time suboptimal control strategy, MPSP has also proven useful as an initialization tool for the nonlinear two-point boundary value problems arising from Pontryagin’s Maximum Principle (PMP). Since MPSP can generate a feasible trajectory with small terminal error at relatively low computational cost, it offers a high-quality initial guess that improves the robustness and convergence of shooting-based methods. This advantage is particularly significant on Lie groups, where such boundary value problems are often highly sensitive to initialization.

More broadly, the landscape of nonlinear optimal control methods reflects a fundamental trade-off between computational efficiency and geometric fidelity. Approaches based on PMP typically require the solution of a nonlinear two-point boundary value problem, making them computationally intensive and sensitive to the initial guess. Methods such as Differential Dynamic Programming (DDP) and iterative Linear Quadratic Regulator (iLQR) avoid explicit boundary value solvers, but rely on repeated linearizations and backward Riccati recursions, which may reduce robustness in strongly nonlinear regimes. In contrast, the proposed MPSP framework on Lie groups avoids both costate propagation and Riccati recursion, resulting in a computationally efficient closed-form control update. Moreover, by showing that the resulting update is equivalent to a Sequential Quadratic Programming (SQP) step, the method admits a rigorous optimization-based interpretation and inherits improved convergence properties. Most importantly, the proposed framework is developed intrinsically on Lie groups, thereby preserving the underlying geometric structure without resorting to local parameterizations.
Although some advances have been made in the geometric formulation of optimal control on matrix Lie groups, the majority of existing works are primarily concerned with theoretical development rather than computational efficiency. Methods based on Pontryagin’s Maximum Principle or discrete variational principles typically lead to two-point boundary value problems or iterative schemes that require backward–forward passes, making them computationally intensive and sensitive to initialization. Even in approaches that extend trajectory optimization techniques to Lie groups, the resulting algorithms often involve repeated linearizations and Riccati-type recursions. As a result, the literature lacks computationally efficient, structure-preserving algorithms that are explicitly designed for real-time optimal control on Lie groups, highlighting a critical gap that this work aims to address. Our earlier work on robust geometric control of a variable‑pitch quadrotor (VPQ) using the Super Twisting Algorithm \cite{krishna2022robust} focused on disturbance rejection and inverted flight capability. Building on that robustness‑oriented foundation, the present paper extends the MPSP framework to systems evolving on Lie groups, addressing optimality and computational efficiency. Many robotic, aerospace, and underwater vehicles can be modeled as simple mechanical systems (SMS), and intrinsic nonlinear control designs for such systems have been studied in \cite{FB-ADL:04,maithripala2006almost,maithripala2015intrinsic}.

The contributions of this work are summarized below:
\begin{enumerate}
    \item An intrinsic formulation of Model Predictive Static Programming for discrete-time systems evolving on Lie groups, based on left-trivialized variations and Lie algebra representations.
    \item A closed-form control update law that avoids both costate propagation and Riccati recursion, leading to reduced computational complexity.
    \item A theoretical connection between the proposed update law and Sequential Quadratic Programming, providing insight into convergence properties.
    \item The MPSP algorithm is compared with another algorithm iterative linear quadratic regulator (iLQR) and compared through numerical simulations.
    \item Application of the proposed algorithm to optimal flipping maneuvers of a variable‑pitch quadrotor (VPQ) and a single‑main‑rotor helicopter (SMRH). Unlike VPQ, where a single rigid‑body assumption suffices, the SMRH requires explicit modeling of rotor–fuselage coupling, motivating the extension of MPSP to handle actuator dynamics and parametric uncertainty.
    \item Derivation of continuous‑time necessary and sufficient optimality conditions for these systems, leading to TPBVPs that are used to benchmark the proposed MPSP approach through numerical simulations.
\end{enumerate}
The key novelty of this work lies in showing that MPSP can be formulated intrinsically on Lie groups using left-trivialized variations, leading to a computational structure that avoids both costate propagation and Riccati recursions. Furthermore, we establish a direct equivalence between the resulting update law and a Sequential Quadratic Programming (SQP) step, providing theoretical convergence guarantees that have not been previously established for MPSP in geometric settings. This work complements earlier robustness‑focused results on VPQ control \cite{krishna2022robust} by addressing optimality and computational efficiency through Lie‑group MPSP. The combination of these approaches lays the groundwork for future integration of robustness and optimality in real‑time control of nonlinear aerospace systems.

The rest of the paper is organized as follows. Section \ref{Sec:Mathematical_Preliminaries} reviews the relevant mathematical preliminaries. Section \ref{Sec:MPSP_on_Lie_groups} presents the proposed MPSP algorithm on Lie groups. Section \ref{Sec:VPQ_SMRH} discusses its application to VPQ and SMRH attitude maneuvers, while Section \ref{Sec:VPQ_SMRH_CT} provides the continuous-time optimality conditions. Numerical comparisons are presented in Section \ref{Sec:Numerical_simulation}, and concluding remarks are given in Section \ref{Sec:Conclusion}.

\section{Mathematical Preliminaries}\label{Sec:Mathematical_Preliminaries}

The geometric framework adopted in this section follows the intrinsic formulation from~\cite{FB-ADL:04}, which provides a standard left‑trivialized representation of simple mechanical systems on Lie groups. We summarize only the definitions and notation required for the development of the proposed Lie‑group MPSP algorithm. Throughout the paper, we use the notation of~\cite{FB-ADL:04}. The set of real numbers is represented by $\mathbb{R}$. The set of non-negative real numbers and the set of strictly positive real numbers are denoted by $
\mathbb{R}_{\ge 0}$ and $\mathbb{R}_{> 0}$ respectively. The set of real $n$ vectors is denoted by $\mathbb{R}^n$ and the set of $m \times n$ real matrices is denoted by $\mathbb{R}^{m \times n}$. The set of $m \times n$ zero matrix is denoted by $\mathbf{0}_{m \times n}$. The identity matrix of order $n$ is denoted by $\mathbf{I}_n$. The operations $\operatorname{det}\left(.\right)$ and $\operatorname{tr}\left(.\right)$ represent the determinant and trace of a matrix respectively. The operation $\Vert .\Vert$ is used to represent the Euclidean norm. 

Consider a smooth manifold $Q$ and let a point $q \in Q$, then $\mathrm{T}_qQ$ represents the tangent space at $q$. Let $X \in \Gamma^{\infty}\left(\mathrm{T}Q\right)$ be a smooth vector field, then the integral curve for $X$ is a curve $\gamma : I \to Q$ at $q_0 \in Q$ satisfying $\gamma'\left(t\right)=X\left(\gamma \left(t\right)\right),~~ t \in I$. Now the flow of $X$ is defined as $\Phi^{X}_t\left(q_0\right)=\gamma\left(t\right)$.  For $X \in \Gamma^{\infty}\left(\mathrm{T}Q\right)$ and $f \in C^{\infty}\left(Q\right)$ be a smooth function on Q, then Lie derivative with respect to $X$ and denote by $\mathcal{L}_X : C^{\infty}\left(Q\right)\to \mathbb{R}$ defined as $\mathcal{L}_Xf\left(q\right)=\left\langle \operatorname{d}f\left(q\right);X\left(q\right) \right\rangle$, where $\operatorname{d}f\left(q\right)\in \mathrm{T}^{*}_qQ $ and $\mathrm{T}^{*}_qQ$ represents the dual of $\mathrm{T}_qQ$. For every smooth vector fields $X, Y \in \Gamma^{\infty}\left(\mathrm{T}Q\right)$, a vector field $\left[X,Y\right]\in \Gamma^{\infty}\left(\mathrm{T}Q\right)$ that satisfies $\mathcal{L}_{\left[X,Y\right]}f=\mathcal{L}_X \mathcal{L}_Yf-\mathcal{L}_Y\mathcal{L}_Xf,~~f\in C^{\infty}\left(Q\right)$. An affine connection on $Q$ is a mapping $\left(X,Y\right)\to \nabla_X Y \in \Gamma^{\infty}\left(\mathrm{T}Q\right) $, for $X,Y \in \Gamma^{\infty}\left(\mathrm{T}Q\right)$,
	\begin{enumerate}
		\item the mapping $\left(X,Y\right)\to \nabla_X Y$ is $\mathbb{R}-$bilinear.
		\item $\nabla_{fX}Y=f\nabla_X Y,~~f\in C^{\infty}\left(Q\right)$.
		\item $\nabla_X\left(fY\right)=f\nabla_XY+\left(\mathcal{L}_Xf\right)Y$.
	\end{enumerate}
	A Riemannian manifold is denoted by the $2-$ tuple $\left(Q,\mathbb{G}\right)$, where $Q$ is a smooth connected manifold and $\mathbb{G}$ is the metric on $Q$. The flat map $\mathbb{G}^{\flat}:\mathrm{T}_qQ \to \mathrm{T}^{*}_qQ$ is given by $\mathbb{G}\left(v_1,v_2\right)=\left\langle \mathbb{G}^\flat\left(v_1\right);v_2\right\rangle$ for $v_1,v_2 \in \mathrm{T}_qQ$ and the sharp map is its dual $\mathbb{G}^\sharp:\mathrm{T}^{*}_qQ\to \mathrm{T}_qQ$ and given by $\mathbb{G}^{-1}\left(w_1,w_2\right)=\left\langle \mathbb{G}^{\sharp}\left(w_1\right),w_2 \right\rangle$ where $w_1,w_2 \in \mathrm{T}^{*}_qQ$. For a Riemannian metric $\mathbb{G}$, there exists a unique affine connection $\stackrel{\mathbb{G}}{\nabla}$, the Levi-Civita connection satisfying
	\begin{enumerate}
		\item $\mathcal{L}_Z \left(\mathbb{G}\left(X,Y\right)\right)=\mathbb{G}\left(\stackrel{\mathbb{G}}{\nabla}_ZX,Y\right)+\mathbb{G}\left(X,\stackrel{\mathbb{G}}{\nabla_ZY}\right)$.
		\item $\stackrel{\mathbb{G}}{\nabla}_X Y-\stackrel{\mathbb{G}}{\nabla}_Y X=\left[X,Y\right]$.
	\end{enumerate}
	Let $\left(G,\star\right)$ be a Lie group where $\star$ represents the binary operation. Let $\mathfrak{g}$ be its Lie algebra and $e$ be the identity element of $G$. The adjoint operator $\operatorname{ad}_\xi : \mathfrak{g} \to \mathfrak{g}$ is defined as $\operatorname{ad}_{\xi}\eta=\left[\xi,\eta\right],~~\forall \xi,\eta \in \mathfrak{g}$, where $\left[ \cdot,\cdot \right]$ is the Lie bracket operation on $\mathfrak{g}$. Next, defining the left translation map $\mathscr{L}_g: G\to G,~~g\in G$ as $\mathscr{L}_g\left(h\right)=g \star h,~~h \in G$. Similarly the right translation map $\mathscr{R}_g: G\to G,~~g\in G $ is defined as $\mathscr{R}_g\left(h\right)=h \star g$. A vector field $X$ on $G$ is left-invariant if $X\left(g \star h\right)=\mathrm{T}_h\mathscr{L}_g\left(X\left(h\right)\right),~~\forall g,h \in G$, where $\mathrm{T}_h\mathscr{L}_g: \mathrm{T}_hG \to \mathrm{T}_{\mathscr{L}_g\left(h\right)}G$ represents the tangent map to $\mathscr{L}_g$ at $h$. From the above definition, it can be inferred that left-invariant vector fields are identified by their value at the identity via the equality $X\left(g\right)=\mathrm{T}_e\mathscr{L}_g\xi$, where $\mathrm{T}_eG \in \xi=X\left(e\right)$. Thus a left-invariant vector field can be defined by the map $g \mapsto \mathrm{T}_e\mathscr{L}_g\left(\xi\right)$. A left-invariant vector field on $G$ is defined by $\xi\left(e\right)=\xi$. Given $g \in G$, define the conjugation map $\mathrm{I}_g:G \to G$ as $h \mapsto g \star h \star g^{-1}$, i.e, $\mathrm{I}_g=\mathscr{L}_g \circ \mathscr{R}_{g^{-1}}$. Consider the tangent map to $\mathrm{I}_g$ at $e$, $\mathrm{T}_e \mathrm{I}_g$ and is called as adjoint map $\operatorname{Ad}_g :\mathfrak{g} \to \mathfrak{g}$ defined as $\mathrm{I}_g \left(\operatorname{exp}\left(\xi\right)\right) = \operatorname{exp}\left(\operatorname{Ad}_g \left(\xi\right)\right)$ and $\left.\frac{\operatorname{d}}{\operatorname{dt}}\right|_{t=0}\operatorname{Ad}_{\operatorname{exp}\left(t\xi\right)}\left(\eta\right)=\operatorname{ad}_{\xi}\eta$. For $g \in G$, the tangent map $\mathrm{T}_e\mathscr{L}_g:\mathrm{T}_eG\to \mathrm{T}_gG$ is a natural isomorphism between $\mathrm{T}_eG$ and $\mathrm{T}_gG$, i.e, $\mathrm{T}_gG \simeq \mathrm{T}_eG$. Thus there also exists an isomorphism between $\mathrm{T}G$ and $G \times \mathrm{T}_eG$ given by $v_g \mapsto \left(g,\mathrm{T}_g\mathscr{L}_{g^{-1}}\left(v_g\right)\right)$ or in other words $\mathrm{T}G \simeq G \times \mathrm{T}_eG$.
	The set $\mathrm{GL}\left(n;\mathbb{R}\right)$ called the real general linear group which represents invertible $n \times n$ matrices with real entries is a Lie group with respect to the operation of matrix multiplication. The identity element is $\mathbf{I}_n$ and the inverse element to $A \in \mathrm{GL}\left(n;\mathbb{R}\right)$ is $A^{-1}$. A matrix Lie group is a subgroup of $\mathrm{GL}\left(n;\mathbb{R}\right)$. Let $\mathfrak{gl}\left(n;\mathbb{R}\right)$ be the Lie algebra of $\mathrm{GL}\left(n;\mathbb{R}\right)$ and we have $\mathfrak{gl}\left(n;\mathbb{R}\right)\simeq \mathbb{R}^{n \times n}$. Consider $A \in \mathfrak{gl}\left(n;\mathbb{R}\right)$, then for all $g \in G$, we have $\mathrm{T}_{\mathbf{I}_n}\mathscr{L}_g\left(A\right)=gA \in \mathrm{T}_gG$. A Riemannian metric $\mathbb{G}$ on a Lie group $\left(G,\star\right)$ is left-invariant if , $	\mathbb{G}\left(g\right)\cdot\left(X_g,Y_g\right)=\mathbb{G}\left(h \star g\right)\cdot\left(\mathrm{T}_g\mathscr{L}_h\left(X_g\right),\mathrm{T}_g\mathscr{L}_h\left(Y_g\right)\right)$ for all $g,h \in G$ and $X_g,Y_g \in \mathrm{T}_gG$. An inner product $\mathbb{I}$ on $\mathfrak{g}$ determines a smooth left-invariant Riemannian metric $\mathbb{G}_{\mathbb{I}}$ on $G$ via left translation $	\mathbb{G}_{\mathbb{I}}\left(g\right)\cdot \left(X_g,Y_g\right)=\mathbb{I}\left(\mathrm{T}_g\mathscr{L}_{g^{-1}}\left(X_g\right),\mathrm{T}_g\mathscr{L}_{g^{-1}}\left(Y_g\right)\right)$. An affine connection on $G$ is left-invariant if and only if the co-variant derivative of any two left-invariant if and only if the co-variant derivative of any two left-invariant vector fields on $G$ is a left-invariant vector field on $G$. Also given a left-invariant affine connection $\nabla$, there exists a unique bilinear map $B : \mathfrak{g} \times \mathfrak{g} \to \mathfrak{g}$ such that for all $\xi,\eta \in \mathfrak{g}$, we have $\nabla_{\xi_L}\eta_L= \left(B\left(\xi,\eta\right)\right)_L$. The converse of the above statement is also true. Let $\mathbb{I}$ be an inner product on $\mathfrak{g}$ and let $\mathbb{G}_{\mathbb{I}}$ be its associated left-invariant Riemannian metric. The Levi-Civita connection $\stackrel{\mathbb{G}_{\mathbb{I}}}{\nabla}$ induced by $\mathbb{G}_{\mathbb{I}}$ is left-invariant and the corresponding bilinear map, denoted by $\stackrel{\mathfrak{g}}{\nabla}:\mathfrak{g}\times \mathfrak{g} \to \mathfrak{g}$ is given by
	\begin{equation}\label{Eq:Bilinear_map_eqn}
		\stackrel{\mathfrak{g}}{\nabla}_\xi\eta=\frac{1}{2}\left[\xi,\eta\right]-\frac{1}{2}\mathbb{I}^{\sharp}\left(\operatorname{ad}^{*}_\xi\mathbb{I}^{\flat}\left(\eta\right)+\operatorname{ad}^{*}_\eta\mathbb{I}^{\flat}\left(\xi\right)\right)
	\end{equation}
	where $\operatorname{ad}^{*}_\xi:\mathfrak{g}^{*}\to \mathfrak{g}$ is the dual operator of $\operatorname{ad}_\xi$ defined by $\left\langle\operatorname{ad}^{*}_\xi\alpha,\eta \right\rangle=\left\langle \alpha,\left[\xi,\eta\right]\right\rangle$ for all $\alpha \in \mathfrak{g}^{*}$. Next, we introduce the notion of a simple mechanical system whose configuration space is a Lie group, which provides the geometric setting for our optimal control formulation. In this context, the configuration of the system evolves on a finite‑dimensional Lie group $G$, while the kinetic energy is defined by a left‑invariant Riemannian metric induced by an inner product on the Lie algebra $\mathfrak{g}$. External and control forces enter through a fixed set of left‑invariant covectors, so that the dynamics can be written intrinsically in body coordinates without resorting to local parameterizations. This definition is standard in geometric mechanics and captures many rigid‑body models of interest, thereby allowing us to exploit the Lie‑group structure systematically in the design and analysis of the MPSP algorithm.

	\begin{definition}\cite{FB-ADL:04}
		A simple mechanical control system (SMS) on a Lie group is a 4-tuple $\left(G,\mathbb{I},\mathcal{F},\mathcal{U}\right)$, where
		\begin{enumerate}
			\item $G$ is an $n-$ dimensional Lie group, defining the configuration space.
			\item $\mathbb{I}$ is an inner product on the Lie algebra $\mathfrak{g}$, defining the kinetic energy metric via left translation.
			\item $\mathcal{F}=\left\lbrace f^{1},\dots,f^{m}\right\rbrace$ is a collection of covectors in $\mathfrak{g}^{*}$, defining a collection of left-invariant control forces.
			\item $\mathcal{U} \subset \mathbb{R}^n$
		\end{enumerate}
	\end{definition}
	Let $\gamma:\mathbb{R}_{\ge 0}\to G$ be the controlled trajectory and define the body velocity as the curve $t \mapsto v\left(t\right)=\mathrm{T}_{\gamma\left(t\right)}\mathscr{L}_{\gamma\left(t\right)}\left(\gamma'\left(t\right)\right)\in \mathfrak{g}$. The controlled equations of motion are
	\begin{equation}\label{Eq:Dynamics_SMS_on_Lie_Group}
		\begin{split}
			\dot{\gamma}\left(t\right)&=\mathrm{T}_e\mathscr{L}_{\gamma\left(t\right)}\left(v\left(t\right)\right)\\
			\dot{v}\left(t\right)-\mathbb{I}^{\sharp}\left(\operatorname{ad}^{*}_{v\left(t\right)}\mathbb{I}^{\flat}\left(v\left(t\right)\right)\right)&= \mathbb{I}^{\sharp}\left(f\left(t,\gamma\left(t\right),v\left(t\right)\right)\right)
		\end{split}
	\end{equation}
	where state feed back force $t \mapsto f\left(t,\gamma\left(t\right),v\left(t\right)\right)\in \mathfrak{g}^{*}$ is defined by control input $u^a: \mathbb{R}_{\ge 0}\times G \times \mathfrak{g} \to \mathbb{R}^n$ by the expression
	\begin{equation*}
		f\left(t,g,v\right)=\sum_{a=1}^{n}f^au^a\left(t,g,v\right)
	\end{equation*}
	Next, we develop the MPSP algorithm for systems evolving on Lie groups, formulated entirely in terms of left‑trivialized states and intrinsic Lie‑algebra variations. Building on the geometric preliminaries and the simple mechanical system model, we derive a discrete‑time deviation dynamics on the Lie algebra that captures how perturbations in the control history affect the terminal state on the manifold. This structure is then used to construct, at each iteration, a static quadratic program in the control variations, where the cost penalizes control effort and the terminal state is enforced as a hard equality constraint through the terminal sensitivity matrices. The resulting subproblem admits a closed‑form solution, avoiding Riccati recursions and costate propagation, and yields an efficient update rule that respects the underlying Lie‑group geometry while enabling real‑time implementation.
	\section{Model Predictive Static Programming on Lie Groups}\label{Sec:MPSP_on_Lie_groups}

    Let $G$ be a finite-dimensional Lie group with identity $e$ and Lie algebra
	$\mathfrak g := \mathrm{T}_eG$. For $g\in G$ denote the left-translation by
	$\mathscr{L}_g(h):=gh$ and its differential by $\mathrm{T}_h\mathscr{L}_g:\mathrm{T}_hG\to \mathrm{T}_{gh}G$.
	The \emph{left Maurer-Cartan form} $\omega^\mathscr{L}\in\Omega^1(G;\mathfrak g)$ is defined as $\omega^\mathscr{L}_g(\xi_g):=\mathrm{T}_g\mathscr{L}_{g^{-1}}\xi_g,\forall\,\xi_g\in \mathrm{T}_gG$, \cite{lee2003smooth} .Thus, $\omega^\mathscr{L}_g$ identifies a tangent vector at $g$ with an element of $\mathfrak g$ by translating it to the identity through left multiplication. Consider a smooth two-parameter family $g:\mathcal U\subset\mathbb R^2\to G$,
	$(t,\epsilon)\mapsto g(t,\epsilon)$, and define the $\mathfrak g$-valued maps
	\begin{equation}\label{eq:v_eta_eps}
		v(t,\epsilon):=\omega^\mathscr{L}_{g(t,\epsilon)}\!\left(\frac{\partial g}{\partial t}(t,\epsilon)\right),
		\qquad
		\eta(t,\epsilon):=\omega^\mathscr{L}_{g(t,\epsilon)}\!\left(\frac{\partial g}{\partial \epsilon}(t,\epsilon)\right).
	\end{equation}
	The map $v(t,\epsilon)\in\mathfrak g$ is the \emph{left-trivialized (body) velocity} along $t$, while $\eta(t,\epsilon)\in\mathfrak g$ is the \emph{left-trivialized variation} along $\epsilon$. In particular, with $g(t):=g(t,0)$ we write
	\begin{equation}\label{eq:v_eta_nominal}
		v(t):=v(t,0)=\omega^\mathscr{L}_{g(t)}(\dot g(t))\in\mathfrak g,\qquad
		\eta(t):=\eta(t,0)=\omega^\mathscr{L}_{g(t)}(\delta g(t))\in\mathfrak g,
	\end{equation}
	where $\delta g(t):=\left.\frac{\partial g}{\partial\epsilon}(t,\epsilon)\right|_{\epsilon=0}$.
	Moreover, define the variation of the body velocity by
	\begin{equation}\label{eq:delta_v_def}
		\delta v(t):=\left.\frac{\partial}{\partial\epsilon}v(t,\epsilon)\right|_{\epsilon=0}\in\mathfrak g.
	\end{equation}
	The Maurer-Cartan structure equation for $\omega^\mathscr{L}$ reads
	\begin{equation}\label{eq:MC_structure_eq}
	\operatorname{	d}\omega^\mathscr{L}+\tfrac12[\omega^\mathscr{L},\omega^\mathscr{L}]=0,
	\end{equation}
	where $[\omega^\mathscr{L},\omega^\mathscr{L}](X,Y):=[\omega^\mathscr{L}(X),\omega^\mathscr{L}(Y)]$ for vector fields $X,Y$ on $G$ and
	$[\cdot,\cdot]$ is the Lie bracket on $\mathfrak g$. Pulling back \eqref{eq:MC_structure_eq} by $g(\cdot,\cdot)$ and setting
	$\alpha:=g^*\omega^\mathscr{L}\in\Omega^1(\mathcal U;\mathfrak g)$ yields
	\begin{equation}\label{eq:MC_pullback_eq}
		\operatorname{d}\alpha+\tfrac12[\alpha,\alpha]=0.
	\end{equation}
	Since $\alpha(\partial_t)=v(t,\epsilon)$ and $\alpha(\partial_\epsilon)=\eta(t,\epsilon)$ by \eqref{eq:v_eta_eps}, evaluating \eqref{eq:MC_pullback_eq} on the commuting coordinate vector fields $(\partial_\epsilon,\partial_t)$ gives
	\begin{align}\label{eq:MC_eval_plane}
		0
		&=(\operatorname{d}\alpha)(\partial_\epsilon,\partial_t)+\tfrac12[\alpha,\alpha](\partial_\epsilon,\partial_t)\nonumber=\partial_\epsilon\!\big(\alpha(\partial_t)\big)-\partial_t\!\big(\alpha(\partial_\epsilon)\big)
		-\alpha([\partial_\epsilon,\partial_t])
		+\tfrac12\Big([\alpha(\partial_\epsilon),\alpha(\partial_t)]-[\alpha(\partial_t),\alpha(\partial_\epsilon)]\Big)\nonumber\\
		&=\partial_\epsilon v(t,\epsilon)-\partial_t\eta(t,\epsilon)+[\eta(t,\epsilon),v(t,\epsilon)].
	\end{align}
	Using antisymmetry $[\eta,v]=-[v,\eta]$, \eqref{eq:MC_eval_plane} becomes
	\begin{equation}\label{eq:delta_v_identity_eps}
		\partial_\epsilon v(t,\epsilon)=\partial_t\eta(t,\epsilon)-[v(t,\epsilon),\eta(t,\epsilon)].
	\end{equation}
	Finally, evaluating \eqref{eq:delta_v_identity_eps} at $\epsilon=0$ and invoking \eqref{eq:delta_v_def} yields
	\begin{equation}\label{eq:delta_v_identity}
		\partial v(t)=\dot\eta(t)-[v(t),\eta(t)].
	\end{equation}
	Rearranging \eqref{eq:delta_v_identity} gives the intrinsic variation relation
	\begin{equation}\label{eq:eta_dot_final}
		\dot\eta(t)=[v(t),\eta(t)]+\partial v(t).
	\end{equation}
	Equivalently, with $\operatorname{ad}_v:\mathfrak g\to\mathfrak g$ defined by $\operatorname{ad}_v(\eta):=[v,\eta]$, we have
	\begin{equation}\label{eq:eta_dot_ad}
		\dot\eta(t)=\operatorname{ad}_{v(t)}\eta(t)+\partial v(t).
	\end{equation}
	Identity \eqref{eq:eta_dot_ad} is purely geometric and does not require any metric or invariance assumptions; it
	captures the non-commutation between time differentiation and variation under left trivialization. To implement MPSP, we represent Lie-algebra and dual Lie-algebra quantities as column vectors.	Let $\{e_1,\dots,e_n\}$ be a basis of $\mathfrak g$ and $\{e^1,\dots,e^n\}$ its dual basis of $\mathfrak g^*$.
	Then $\eta$ and $v$ admit the expansions $\eta=\eta^i e_i, v=v^i e_i,$ and a control input $u\in\mathfrak g^*$ can be written as $u=u_i e^i,$ where Einstein summation is used. Define the corresponding column vectors
	\begin{equation}\label{eq:column_vectors}
		\underline{\eta} := \begin{bmatrix}\eta^1&\cdots&\eta^n\end{bmatrix}^\top,\quad
		\underline{v} := \begin{bmatrix}v^1&\cdots&v^n\end{bmatrix}^\top,\quad
		\underline{u} := \begin{bmatrix}u_1&\cdots&u_n\end{bmatrix}^\top.
	\end{equation}
	Moreover, let the Lie bracket be encoded by structure constants $c^k_{ij}$ defined by $	[e_i,e_j]=c^k_{ij}\,e_k.$ Then, for $v=v^ie_i$ and $\eta=\eta^je_j$, we have $	[v,\eta]=v^i\eta^j c^k_{ij}\,e_k$. The term $	\underline{[v,\eta]}$ has components $([v,\eta])^k=v^i\eta^j c^k_{ij}$. Thus, for fixed $v$, the map $\eta\mapsto [v,\eta]$ is linear in $\eta$ and can be represented by a matrix constructed from $v$ and the structure constants. The dynamics of the system is given by 
	\begin{equation}\label{eq:v_dyn}
		\dot v(t)=\mathbb I^{\sharp}\!\Big(\operatorname{ad}_v^{*}\big(\mathbb I^{\flat}(v(t))\big)\Big)+\mathbb I^{\sharp}\big(u(t)\big),
	\end{equation}
	where $\mathbb I^{\flat}:\mathfrak g\to\mathfrak g^*$ is an invertible inertia operator and
	$\mathbb I^{\sharp}:=(\mathbb I^{\flat})^{-1}:\mathfrak g^*\to\mathfrak g$.
	The coadjoint operator $\operatorname{ad}_v^{*}:\mathfrak g^*\to\mathfrak g^*$ is defined by $\langle \operatorname{ad}_v^{*}\mu,\xi\rangle := \langle \mu,[\xi,v]\rangle, \forall\,\mu\in\mathfrak g^{*},~\forall\,\xi\in\mathfrak g.$ Let $u_0(\cdot)$ be the control input at the previous iteration and $v_0(\cdot)$ the corresponding velocity trajectory:
	\begin{equation}\label{eq:v0_dyn}
		\dot v_0(t)=\mathbb I^{\sharp}\!\Big(\operatorname{ad}_{v_0}^{*}\big(\mathbb I^{\flat}(v_0(t))\big)\Big)+\mathbb I^{\sharp}\big(u_0(t)\big).
	\end{equation}
	If $u_0$ does not satisfy the control objective, we update the control input and the velocity as
	\begin{equation}\label{eq:inc_def}
		u(t)=u_0(t)+\partial u(t),\qquad v(t)=v_0(t)+\partial v(t),
	\end{equation}
	so that $\dot v(t)=\dot v_0(t)+\partial\dot v(t)$.
	Substituting \eqref{eq:inc_def} into \eqref{eq:v_dyn} and subtracting \eqref{eq:v0_dyn} yields
	\begin{equation}\label{eq:dv_exact}
		\partial\dot v(t)=\Big(F(v_0(t)+\partial v(t))-F(v_0(t))\Big)+\mathbb I^{\sharp}\big(\partial u(t)\big),
	\end{equation}
	where $	F(v):=\mathbb I^{\sharp}\!\Big(\operatorname{ad}_v^{*}\big(\mathbb I^{\flat}(v)\big)\Big).$ Next, we discuss the first-order expansion of $F$ at $v_0$. Let $w\in\mathfrak g$ and consider $v=v_0+\epsilon w$. Since $\mathbb I^{\flat}$ is linear and
	$(v,\mu)\mapsto \operatorname{ad}_v^{*}\mu$ is bilinear, we have
	\begin{align}\label{eq:adstar_expand}
		\operatorname{ad}^{*}_{v_0+\epsilon w}\!\Big(\mathbb I^{\flat}(v_0+\epsilon w)\Big)
		&=\big(\operatorname{ad}_{v_0}^{*}+\epsilon\operatorname{ad}_{w}^{*}\big)\Big(\mathbb I^{\flat}(v_0)+\epsilon\mathbb I^{\flat}(w)\Big)\nonumber=\operatorname{ad}_{v_0}^{*}\big(\mathbb I^{\flat}(v_0)\big)
		+\epsilon\,\operatorname{ad}_{v_0}^{*}\big(\mathbb I^{\flat}(w)\big)
		+\epsilon\,\operatorname{ad}_{w}^{*}\big(\mathbb I^{\flat}(v_0)\big)
		+\mathcal O(\epsilon^2).
	\end{align}
		Applying $\mathbb I^{\sharp}$ gives
	\begin{equation}\label{eq:F_expand}
		F(v_0+\epsilon w)=F(v_0)
		+\epsilon\,\mathbb I^{\sharp}\!\Big(
	\operatorname{ad}_{v_0}^{*}\big(\mathbb I^{\flat}(w)\big)
		+\operatorname{ad}_{w}^{*}\big(\mathbb I^{\flat}(v_0)\big)
		\Big)
		+\mathcal O(\epsilon^2).
	\end{equation}
	Hence, the Fr\'echet derivative of $F$ at $v_0$ in the direction $w$ is
	\begin{equation}\label{eq:DF_general}
		(DF)_{v_0}[w]
		=
		\mathbb I^{\sharp}\!\Big(
	\operatorname{ad}_{v_0}^{*}\big(\mathbb I^{\flat}(w)\big)
		+\operatorname{ad}_{w}^{*}\big(\mathbb I^{\flat}(v_0)\big)
		\Big),
	\end{equation}
	Using \eqref{eq:F_expand} in \eqref{eq:dv_exact} with $w=\partial v$ and neglecting $\mathcal O(\|\partial v\|^2)$ yields
	\begin{equation}\label{eq:dv_lin_general}
		\partial\dot v(t)=
		\mathbb I^{\sharp}\!\Big(
	\operatorname{ad}_{v_0(t)}^{*}\big(\mathbb I^{\flat}(\partial v(t))\big)
		+\operatorname{ad}_{\partial v(t)}^{*}\big(\mathbb I^{\flat}(v_0(t))\big)
		\Big)
		+\mathbb I^{\sharp}\big(\partial u(t)\big).
	\end{equation}
Let $\{e_i\}_{i=1}^n$ be a basis of $\mathfrak g$ with structure constants $[e_i,e_j]=c^k_{ij}e_k$, and let $\{e^i\}_{i=1}^n$ be the dual basis of $\mathfrak g^{*}$.
Consider the components of $v , \mu \in \mathfrak{g}$ is represented as $v=v^ie_i$, $\mu=\mu_ie^i$. Then, the matrix representation of $\operatorname{ad}_v^{*}\mu$ can be defined as $\Big[\operatorname{ad}_v^{*}\mu\Big]_k = v^i c^{\,j}_{k i}\mu_j$. Define the matrix $C(\underline v)\in\mathbb R^{n\times n}$ by $C(\underline v)_{k j}:=v^i c^{\,j}_{k i}$ so that $\Big[\operatorname{ad}_v^{*}\mu\Big]=C(\underline v)\,\underline\mu$. Let $\Big[\mathbb{I}^{\flat}\Big]$ be the matrix of $\mathbb I^{\flat}$ in these bases, then the matrix representation of  $\mathbb{I}^{\sharp}$ is given by $\Big[\mathbb{I}^{\sharp}\Big]=\Big[\mathbb{I}^{\flat}\Big]^{-1}$. The dynamics of $\partial v\left(t\right)$ given in \eqref{eq:dv_lin_general} can be written in the matrix form as
\begin{equation}\label{eq:dv_lin_coord}
	\underline{\partial\dot v}
	=A_{v\left(t\right)}+\Big[\mathbb{I}^{\flat}\Big]^{-1}\underline{\partial u}.
\end{equation}
where $A_v =-\left[\mathbb{I}^{\sharp}\right]\left[\operatorname{ad}_{v}\right]\left[\mathbb{I}^{\flat}\right]+\left[\operatorname{ad}_{v}\right]$. In the discrete time setting, we have the fixed integration time step $h \in \mathbb{R}$, then the time variable is represented as $t=k h$, where $k=1, \dots, N$ and $N \in \mathbb{N}$ is the number of time steps. The discrete time evolution of the state $\partial X$ is given by
\begin{equation}\label{eq:lin_dev_dyn}
	\partial X_{k+1}=A_k \partial X_k+B_k \partial \underline{u}_k
\end{equation}
where
\begin{equation}\label{Eq:A_k_B_k_def}
	\begin{split}
		A_k & =\begin{bmatrix}
			A_{\eta} \left(\eta_{0,k},v_{0,k}\right)& h\mathbb{I}_n \\
			\mathbf{0}_{n \times n} & \mathbb{I}_n +h A_v\left(v_{0,k},\partial v_k\right)
           \end{bmatrix},~B_k=\begin{bmatrix}
		\mathbf{0}_{n \times n}\\
		h \left[\mathbb{I}^{\flat}\right]^{-1}
	\end{bmatrix}
	\end{split}
\end{equation}
The variation in the final time step $k=N$ can be expressed as
\begin{equation}
	\partial X_N=\prod_{k=N-1}^{1} A_k \partial X_1+\sum_{k=1}^{N-1}G_k \partial \underline{u}_k
\end{equation}
where
\begin{equation}\label{eq:Gk_def}
	\begin{split}
		G_k &=A_{N-1}A_{N-2}\dots A_{k+1}B_k, ~k=1, \dots, N-2\\
		G_{N-1}&=B_k
	\end{split}
\end{equation}
Since in the objective of the problem the end points are fixed, we have the property $\partial X_1 =\mathbf{0}_{2n \times 1}$ and $\partial X_N$ which can be expressed as
\begin{equation}\label{eq:terminal_constraint_mpsp}
	\partial X_N =\sum_{k=1}^{N-1}G_k \partial \underline{u}_k
\end{equation}
It should be noted that $G_k$ can be found recursively as
\begin{equation}\label{Eq:G_k_recursive_iteration}
	\begin{split}
		G_k &= G^{0}_k B_k,\\
		G^{0}_k&=G^{0}_{k+1}A_{k+1},~k=N-2,\dots,1\\
		G^{0}_{N-1}&=\mathbb{I}_{2n}
	\end{split}
\end{equation}
As mentioned earlier MPSP is an iterative procedure. Let $i \in \mathbb{N}$ be the current iteration number. Let $\left\lbrace \underline{u}^{i-1}_k\right\rbrace^{N-1}_{k=1}$ be the optimal control sequence found using MPSP at $\left(i-1\right)^{\text{th}}$ iteration which will not satisfy the hard constraints, i.e, $\partial X^{i-1}_{N} \neq \mathbf{0}_{2 n \times 1}$. We define optimal control sequence $\left\lbrace \underline{u}^{i}_k\right\rbrace^{N-1}_{k=1}$ for the current $i ^{\text{th}}$ iteration calculated as
\begin{equation}
	\underline{u}^{i}_k= \underline{u}^{i-1}_k -\partial \underline{u}^{i}_k
\end{equation}
The objective of MPSP is stated as below
\begin{equation}\label{Eq:Problem_MPSP}
	\begin{aligned}
		& \operatorname{minimize}_{\left\{\partial \underline{u}_k\right\}_{k=1}^{N-1}} \quad J:=\sum_{k=1}^{N-1} \frac{1}{2} \left(\underline{u}^{i-1}_k-\partial \underline{u}^{i}_k\right)^\top R_k \left(\underline{u}^{i-1}_k-\partial \underline{u}^{i}_k\right) \\
		& \text { subject to }\left\{\begin{array}{l}
			\partial X^{i-1}_N =\sum_{k=1}^{N-1}G_k\partial \underline{u}^{i}_k
		\end{array}\right.
	\end{aligned}
\end{equation}
In the problem objective given in \eqref{Eq:Problem_MPSP}, the goal is to minimize the total control effort measured by the deviation of the updated control sequence $\underline{u}^{i}_k$ from the previous iterate $\underline{u}^{i-1}_k$, that is, the norm of the correction $\underline{u}^{i-1}_k - \partial \underline{u}^{i}_k$. This reflects the MPSP philosophy of computing the smallest control adjustment that exactly enforces the desired terminal state, rather than re‑optimizing the entire control history from scratch at each iteration. The terminal constraint is imposed as a hard equality through the linearized endpoint relation, ensuring that the updated control sequence drives the terminal left‑trivialized state deviation to zero while keeping the incremental control cost minimal. The next theorem formalizes this optimization step by deriving the first‑order optimality conditions (in the form of Karush–Kuhn–Tucker conditions) for the resulting quadratic program on Lie groups, and provides the closed‑form expression for the optimal control update used in the MPSP algorithm.
\begin{theorem}\label{Th:MPSP_1}
	Let $\left\lbrace \stackrel{\circ}{\underline{u}_k}^{i}\right\rbrace^{N-1}_{k=1}$ be an optimal control sequence that solves the problem objective \eqref{Eq:Problem_MPSP}, then the necessary conditions for optimality are
	\begin{itemize}
		\item Primal feasibility
		\begin{equation}
			\partial X^{i-1}_N =\sum_{k=1}^{N-1}G_k \partial \stackrel{\circ}{\underline{u}_k}^{i}
		\end{equation}
		\item Optimal control
		\begin{equation}\label{Eq:MPSP_obj_1_control}
			\stackrel{\circ}{\underline{u}_k}^{i}=R^{-1}_k G_k\left(P^{-1}\left(Q-\partial X^{i-1}_N\right)\right)
		\end{equation}
		where $P=-\sum_{k=1}^{N-1}G_k R^{-1}_k G^\top_k$ and $Q=\sum_{k=1}^{N-1}G_k\stackrel{\circ}{\underline{u}_k}^{i}$.
	\end{itemize}
\end{theorem}
\begin{proof}
The optimal control problem stated is a quadratic programming problem which is convex. the optimality conditions are given by the Karush Kuhn Tucker (KKT) conditions. the Lagrangian $L: \mathbb{R}^{\left(N-1\right)m} \times \mathbb{R}^{2n} \to \mathbb{R}$ is defined as
	\begin{equation}
		L\left(\partial \underline{u}^{i}_1, \dots, \partial \underline{u}^{i}_{N-1}, \mu\right)=\frac{1}{2}\sum_{k=1}^{N-1}\left(\underline{u}^{i-1}_{k-1}-\partial \underline{u}^{i}_k\right)^\top R_k \left(\underline{u}^{i-1}_{k-1}-\partial \underline{u}^{i}_k\right)+\mu^\top \left(\partial X^{i-1}_N-\sum_{k=1}^{N-1}G_k \partial \underline{u}^{i}_k\right)
	\end{equation}
where $\mu \in \mathbb{R}^{2n \times 1}$ is the Lagrangian multiplier as dual variable. The KKT conditions are
\begin{enumerate}
	\item Primal feasibility
	\begin{equation*}
			\partial X^{i-1}_N =\sum_{k=1}^{N-1}G_k \partial \stackrel{\circ}{\underline{u}_k}^{i}
	\end{equation*}
\item First order condition
\begin{equation*}
	-R_k\left(\stackrel{\circ}{\underline{u}_k}^{i-1}- \stackrel{\circ}{\partial \underline{u}_k}^{i}\right)-G^\top_k \stackrel{\circ}{\mu}=0, ~k=1, \dots, N-1.
\end{equation*}
\end{enumerate}
Simplifying, we get
\begin{equation}\label{Eq:Control_input_int}
	\partial \stackrel{\circ}{\underline{u}_k}^{i}=R^{-1}_k G^\top_k \stackrel{\circ}{\mu}+\stackrel{\circ}{\underline{u}_k}^{i}
\end{equation}
Substituting the expression of $\partial \stackrel{\circ}{\underline{u}_k}^{i}$ in the primal feasibility equation we get,
\begin{equation*}
	\begin{split}
		\partial X^{i-1}_{N} &= \sum_{k=1}^{N-1}G_k \partial \stackrel{\circ}{\underline{u}_k}^{i}=\sum_{k=1}^{N-1}G_k\left(R^{-1}_k G^\top_k \stackrel{\circ}{\mu}+\stackrel{\circ}{\underline{u}_k}^{i}\right)=\left(\sum_{k=1}^{N-1}G_k R^{-1}_k G^\top_k\right)\stackrel{\circ}{\mu}+\sum_{k=1}^{N-1}\left(G_k \stackrel{\circ}{\underline{u}_k}^{i}\right)=-P \stackrel{\circ}{\mu}+Q
	\end{split}
\end{equation*}
where $P=-\sum_{k=1}^{N-1}G_k R^{-1}_k G^\top_k$ and $Q= \sum_{k=1}^{N-1}\left(G_k \stackrel{\circ}{\underline{u}_k}^{i}\right)$.
Now the dual variable $\stackrel{\circ}{\mu}$ is obtained as
\begin{equation}
\stackrel{\circ}{\mu}=P^{-1}\left(Q-\partial X^{i-1}_N\right)
\end{equation}
Substituting $\stackrel{\circ}{\mu}$ in \eqref{Eq:Control_input_int}, we get the optimal controller as
\begin{equation*}
\partial \stackrel{\circ}{\underline{u}_k}^{i}=R^{-1}_k G^\top_k\left(	P^{-1}\left(Q-\partial X^{i-1}_N\right)\right)+\stackrel{\circ}{\underline{u}_k}^{i}
\end{equation*}
Now the optimal control is given by
\begin{equation}
	\stackrel{\circ}{\underline{u}_k}^{i}=R^{-1}_k G^\top_k\left(P^{-1}\left(Q-\partial X^{i-1}_N\right)\right)
\end{equation}
It can be observed that the updated control history is obtained in closed form. 
\end{proof}

\noindent Next, we intend to minimize the deviation in the control input $\partial \underline{u}^{i}_k$ and the problem objective is stated as below 
\begin{equation}\label{Eq:Problem_MPSP_2}
	\begin{aligned}
		& \operatorname{minimize}_{\left\{\partial \underline{u}_k\right\}_{k=1}^{N-1}} \quad J:=\sum_{k=1}^{N-1} \frac{1}{2} \left(\partial \underline{u}^{i}_k\right)^\top R_k \left(\partial \underline{u}^{i}_k\right) \\
		& \text { subject to }\left\{\begin{array}{l}
			\partial X^{i-1}_N =\sum_{k=1}^{N-1}G_k\partial \underline{u}^{i}_k
		\end{array}\right.
	\end{aligned}
\end{equation}
The next theorem proposes the optimal control sequence $\left\lbrace \stackrel{\circ}{\underline{u}_k}^{i}\right\rbrace^{N-1}_{k=0}$ that will minimize the above problem objective.
\begin{theorem}\label{Th:MPSP_2}
	Let $\left\lbrace \stackrel{\circ}{\underline{u}_k}^{i}\right\rbrace^{N-1}_{k=1}$ be an optimal control sequence that solves the problem objective \eqref{Eq:Problem_MPSP_2}, then the necessary conditions for optimality are
	\begin{itemize}
		\item Primal feasibility
		\begin{equation}
			\partial X^{i-1}_N =\sum_{k=1}^{N-1}G_k \partial \stackrel{\circ}{\underline{u}_k}^{i}
		\end{equation}
		\item Optimal control
		\begin{equation}
			\stackrel{\circ}{\underline{u}_k}^{i}=\stackrel{\circ}{\underline{u}_k}^{i-1}-R^{-1}_k G_kP^{-1}\partial X^{i-1}_N, ~ k=1, \dots, N-1.
		\end{equation}
		where $P=-\sum_{k=1}^{N-1}G_k R^{-1}_k G^\top_k$
		\end{itemize}
\end{theorem}
\begin{proof}
The optimal control problem stated is also a quadratic programming problem which is convex. The optimality conditions are given by the KKT conditions. The Lagrangian  $L: \mathbb{R}^{\left(N-1\right)m} \times \mathbb{R}^{2n} \to \mathbb{R}$ is defined as
	\begin{equation*}
		L\left(\partial \underline{u}^{i}_1, \dots, \partial \underline{u}^{i}_{N-1}, \mu\right)=\frac{1}{2}\sum_{k=1}^{N-1}\left(\partial \underline{u}^{i}_k\right)^\top R_k \partial \underline{u}^{i}_k+\mu^\top \left(\partial X^{i-1}_N-\sum_{k=1}^{N-1}G_k \partial \underline{u}^{i}_k\right)
	\end{equation*}
where $\mu \in \mathbb{R}^{2n \times 1}$ is the Lagrangian multiplier as dual variable. The KKT conditions are
\begin{enumerate}
	\item Primal feasibility
	\begin{equation*}
		\partial X^{i-1}_N =\sum_{k=1}^{N-1}G_k \partial \stackrel{\circ}{\underline{u}_k}^{i}
	\end{equation*}
\item First order condition
\begin{equation*}
	R_k \partial \stackrel{\circ}{\underline{u}_k}^{i}-G^\top_k \stackrel{\circ}{\mu},~k=1, \dots, N-1.
\end{equation*}
\noindent From the above equation, we get
\begin{equation}\label{Eq:Control_input_int_2}
	\partial \stackrel{\circ}{\underline{u}_k}^{i}=R^{-1}_k G^\top_k\stackrel{\circ}{\mu}
\end{equation}
Substituting \eqref{Eq:Control_input_int_2} in the primal feasibility, we get
\begin{equation*}
	\partial X^{i-1}_{N}= \left(\sum_{k=1}^{N-1}G_k R^{-1}_k G^\top_k\right)\stackrel{\circ}{\mu}
\end{equation*} 
The optimal dual variable is obtained as $\stackrel{\circ}{\mu}= P^{-1}\partial X^{i-1}_N$. Substituting $\stackrel{\circ}{\mu}$ in \eqref{Eq:Control_input_int_2}, we get
\begin{equation*}
\stackrel{\circ}{\underline{u}_k}^{i}=\stackrel{\circ}{\underline{u}_k}^{i-1}-R^{-1}_k G_kP^{-1}\partial X^{i-1}_N, ~ k=1, \dots, N-1.
\end{equation*} 
\end{enumerate}
\end{proof}

\noindent We discussed the optimality conditions for the MPSP algorithm in Theorems \ref{Th:MPSP_1} and \ref{Th:MPSP_2}, which characterize, for a fixed iteration, how the control corrections solve the underlying quadratic subproblems and enforce the terminal constraint. However, these results are local in nature: they describe the structure of a single MPSP step but do not guaranty by themselves that the iterated control sequence $\{\underline{u}_k^{\,i-1}\}$ converges to a truly optimal control law as the iteration index $i$ increases. In particular, the previous analysis does not address whether repeated application of the MPSP update drives the terminal deviation to zero and the control history toward a KKT point of the finite‑horizon optimal control problem. In the next section, we close this gap by reformulating the MPSP iteration as a Sequential Quadratic Programming (SQP) scheme applied to the nonlinear program defined by the endpoint map and control‑effort cost, and we establish local convergence of the resulting iterates under standard regularity and second‑order sufficiency assumptions. This provides a rigorous justification for the use of MPSP as an iterative optimal control method on Lie groups, beyond the per‑iteration optimality characterized by Theorems \ref{Th:MPSP_1} and \ref{Th:MPSP_2}.

\section{Convergence of the MPSP algorithm.}

At first, we reformulate the MPSP algorithm as a nonlinear program with control history as the optimizing variable. Let the terminal deviation be $\partial X_N(U)\coloneqq
\begin{bmatrix}\eta_N(U) &      \partial v_N(U)\end{bmatrix}\in\mathbb{R}^{2n}$, where $\eta_k,\partial v_k$ are the left-trivialized deviation variables generated by the dynamics and deviation definitions
used throughout the manuscript. Denote the stacked control history by $U\coloneqq \mathrm{col}(u_1,\dots,u_{N-1})\in\mathbb{R}^{q}, q=(N-1)n,$ and define the endpoint map $H(U)\coloneqq\partial X_N(U)\in\mathbb{R}^{p},p=2n.$ 
At iterate $U^{i-1}$, let the (discrete-time) linear deviation model be
\begin{equation}\label{eq:lin_dev}
	\partial X_{k+1}=A_k\,\partial X_k + B_k\,\partial u_k,\qquad k=1,\dots,N-1,
\end{equation}
and define the terminal sensitivity matrices
\begin{equation}\label{eq:Gk_def_again}
	G_k \;\coloneqq\; A_{N-1}A_{N-2}\cdots A_{k+1}B_k,\qquad k=1,\dots,N-1,
\end{equation}
so that the stacked Jacobian is
\begin{equation}\label{eq:G_stack}
	\mathcal{G}^{\,i-1}\;\coloneqq\;[\,G_1\ \cdots\ G_{N-1}\,]\in\mathbb{R}^{2n\times (N-1)n}.
\end{equation}
With $\partial U^i\coloneqq\mathrm{col}(\partial u_1^i,\dots,\partial u_{N-1}^i)$, the MPSP step enforces the
linearized terminal equality
\begin{equation}\label{eq:linear_terminal}
	\partial X_N^{\,i-1} \;=\; \mathcal{G}^{\,i-1}\,\partial U^i,
\end{equation}
and computes $\partial U^i$ as the unique minimizer of the strictly convex quadratic program
\begin{equation}\label{eq:mpsp_qp_again}
	\min_{\partial U^i}\;\frac12\,(U^{i-1}-\partial U^i)^\top R\,(U^{i-1}-\partial U^i)
	\qquad \text{s.t.}\qquad \partial X_N^{\,i-1}=\mathcal{G}^{\,i-1}\partial U^i.
\end{equation}
where $R\coloneqq\mathrm{blkdiag}(R_1,\dots,R_{N-1})\succ 0$. The control update is then $U^{i}=U^{i-1}-\partial U^i$. The target terminal condition $\partial X_N=0$ can be cast as the equality-constrained nonlinear program
\begin{equation}\label{eq:nlp_terminal}
	\min_{U\in\mathbb{R}^q}\; J(U)\qquad \text{s.t.}\qquad H(U)=0.
\end{equation}
We now state the standing assumptions used in the local convergence analysis of the MPSP algorithm. These conditions formalize the required smoothness of the endpoint map and cost functional, as well as regularity and curvature properties at the optimal solution of the equality‑constrained nonlinear program associated with the MPSP step. In particular, we assume that the endpoint map is twice continuously differentiable with Lipschitz continuous second derivatives in a neighborhood of the optimal control history, that the Jacobian of the terminal constraint has full row rank (ensuring a valid linearization and satisfaction of the LICQ condition), and that a standard second‑order sufficient optimality condition holds via positive definiteness of the reduced Hessian of the Lagrangian on the constraint tangent space. Finally, we require that any globalization or line‑search strategy used in practice accepts asymptotically full steps near the solution. Collectively, these assumptions provide the technical foundation needed to apply classical SQP convergence results to the Lie‑group MPSP scheme.
\paragraph{Assumptions.}
Let $\stackrel{\circ}{U}$ be a KKT point of \eqref{eq:nlp_terminal} with associated multiplier $\stackrel{\circ}{\lambda}\in\mathbb{R}^{2n}$.
Assume:
\begin{enumerate}
	\item[\textbf{(A1)}] (\emph{Smooth endpoint map}) $H$ and $J$ are twice continuously differentiable in a neighborhood
	$\mathcal{N}$ of $\stackrel{\circ}{U}$  and their second derivatives are Lipschitz continuous on $\mathcal{N}$.
	\item[\textbf{(A2)}] (\emph{LICQ / full-row-rank terminal sensitivity})
	\begin{equation}\label{eq:licq_rank}
		\mathrm{rank}\big(\nabla H(\stackrel{\circ}{U} )\big)=2n.
	\end{equation}
	Equivalently, for all $U\in\mathcal{N}$ sufficiently close to $\stackrel{\circ}{U} $,
	$\mathrm{rank}(\mathcal{G}^{\,i-1})=2n$ (hence $\mathcal{G}^{\,i-1}R^{-1}(\mathcal{G}^{\,i-1})^\top\succ0$).
	\item[\textbf{(A3)}] (\emph{Second-order sufficiency})
	Let $\mathcal{L}(U,\lambda)\coloneqq J(U)+\lambda^\top H(U)$ be the Lagrangian of \eqref{eq:nlp_terminal}.
	Then the reduced Hessian is positive definite:
	\begin{equation}\label{eq:sosc}
		d\neq 0,\ \nabla H(\stackrel{\circ}{U})d=0 \quad\Longrightarrow\quad
		d^\top\nabla^2_{UU}\mathcal{L}(\stackrel{\circ}{U},\stackrel{\circ}{\lambda})\,d>0.
	\end{equation}
	\item[\textbf{(A4)}] (\emph{Full steps locally}) Any globalization used in the implementation accepts unit steps
	in $\mathcal{N}$, i.e., the step sizes satisfy $\alpha_i\to 1$.
\end{enumerate}
Next, we present a lemma establishing the full‑row‑rank property of the terminal sensitivity matrix, which plays a central role in the local convergence analysis of the MPSP algorithm. This result shows that, under mild assumptions on the integration step and inertia operator, the stacked sensitivity map from control perturbations to the terminal deviation is surjective onto the terminal state space. In particular, the lemma guarantees that  has rank, ensuring that any small desired correction in the terminal left‑trivialized state can be achieved by an appropriate control variation over the horizon. This property underpins the Linear Independence Constraint Qualification (LICQ) in the nonlinear programming formulation and is therefore instrumental in invoking standard SQP convergence results for the proposed Lie‑group MPSP scheme.
	\begin{lemma}[Full-Row-Rank of Terminal Sensitivity]
		\label{lem:full-rank}
		Let $h > 0$ be the integration step size, and suppose the inertia
		operator $\mathbb{I}^{\flat}: \mathfrak{g} \to \mathfrak{g}^{*}$ is
		invertible with $\mathbb{I}^{\sharp} := (\mathbb{I}^{\flat})^{-1}$.
		Then, for any $N \geq 3$ and any nominal trajectory
		$\{(\eta_{0,k}, v_{0,k})\}_{k=1}^{N-1}$, the stacked terminal
		sensitivity matrix
		\[
		\mathcal{G}^{i-1}
		\;=\; [G_1 \;\cdots\; G_{N-1}]
		\;\in\; \mathbb{R}^{2n \times (N-1)n}
		\]
		satisfies $\mathrm{rank}(\mathcal{G}^{i-1}) = 2n$, and consequently
		$\mathcal{G}^{i-1} R^{-1} (\mathcal{G}^{i-1})^{\top} \succ 0$
		for any $R \succ 0$.
		In particular, Assumption~\textnormal{(A2)} holds for all $N \geq 3$.
	\end{lemma}
	\begin{proof}
		It suffices to show that the $2n \times 2n$ submatrix
		$[G_{N-2} \;\; G_{N-1}]$ has rank $2n$. From the discrete-time
		sensitivity recursion \eqref{Eq:G_k_recursive_iteration}, the last two blocks are
		\begin{equation}
			G_{N-1} = B_{N-1},
			\qquad
			G_{N-2} = A_{N-1}\, B_{N-2}.
		\end{equation}
		Using the block structure of $A_k$ and $B_k$ from \eqref{Eq:A_k_B_k_def},
		namely
	\begin{equation}
		A_k =
		\begin{bmatrix}
			A_{\eta}(\eta_{0,k}, v_{0,k}) & h I_n \\
			0_{n \times n} & I_n + h A_v(v_{0,k})
		\end{bmatrix},
		\qquad
		B_k =
		\begin{bmatrix}
			0_{n \times n} \\
			h\, \mathbb{I}^{\sharp}
		\end{bmatrix},
	\end{equation}
	a direct computation gives
	\begin{equation}
		G_{N-2} = A_{N-1} B_{N-2}
		=
		\begin{bmatrix}
			h^2\, \mathbb{I}^{\sharp} \\
			h\bigl(I_n + h A_v(v_{0,N-1})\bigr)\mathbb{I}^{\sharp}
		\end{bmatrix}.
	\end{equation}
	Concatenating $G_{N-2}$ and $G_{N-1}$ into a $2n \times 2n$ matrix
	yields
	\begin{equation}
		\bigl[G_{N-2} \;\; G_{N-1}\bigr]
		=
		\begin{bmatrix}
			h^2\, \mathbb{I}^{\sharp} & 0_{n \times n} \\
			h\bigl(I_n + h A_v(v_{0,N-1})\bigr)\mathbb{I}^{\sharp}
			& h\, \mathbb{I}^{\sharp}
		\end{bmatrix}.
		\label{eq:two-block-mat}
	\end{equation}
	The matrix \eqref{eq:two-block-mat} is block lower-triangular with
	diagonal blocks $h^2\mathbb{I}^{\sharp}$ and $h\mathbb{I}^{\sharp}$.
	Since $\mathbb{I}^{\flat}$ is invertible and $h > 0$, both diagonal
	blocks are invertible, so
	\begin{equation}
		\det\bigl[G_{N-2} \;\; G_{N-1}\bigr]
		= \det\!\bigl(h^2 \mathbb{I}^{\sharp}\bigr)
		\det\!\bigl(h\, \mathbb{I}^{\sharp}\bigr)
		= h^{3n}\!\bigl(\det \mathbb{I}^{\sharp}\bigr)^2
		\;\neq\; 0.
	\end{equation}
	Hence $[G_{N-2} \;\; G_{N-1}]$ is invertible, which implies
	$\mathrm{rank}(\mathcal{G}^{i-1}) = 2n$ since
	$[G_{N-2} \;\; G_{N-1}]$ is a submatrix of $\mathcal{G}^{i-1}$.
	Consequently, $\mathcal{G}^{i-1} R^{-1}(\mathcal{G}^{i-1})^{\top}
	\succ 0$ for any $R \succ 0$, which is equivalent to (A2).
	\end{proof}
    
   \noindent The next theorem establishes the local convergence of the MPSP algorithm under the regularity and curvature assumptions \textbf{(A1)}–\textbf{(A4)}. Under this framework, the theorem shows that, starting from an initial control history sufficiently close to a KKT point of the finite‑horizon optimal control problem, the sequence of control updates generated by the Lie‑group MPSP iterations is well defined and converges to the optimal control. Moreover, the associated terminal deviation converges to zero at least ‑linearly, thereby providing a rigorous justification for the practical effectiveness and reliability of the proposed MPSP scheme on Lie groups.
\begin{theorem}\label{thm:local_Qlinear_mpsp}
	Under the assumptions \textbf{(A1)}--\textbf{(A4)}, there exists $\varepsilon>0$ such that if $\|U-\stackrel{\circ}{U}\|<\varepsilon$, then:
	(i) the MPSP subproblem \eqref{eq:mpsp_qp_again} is feasible and has a unique solution for all $i$,
	(ii) the iterates $\{U^i\}$ are well-defined and converge to $\stackrel{\circ}{U}$, and
	(iii) the terminal deviation converges to zero $Q$-linearly:
	\begin{equation}\label{eq:qlinear_terminal}
		\|\partial X_N(U^{i})\|\;=\;\|H(U^{i})\|\ \le\ \rho\,\|H(U^{i-1})\| \qquad \text{for all $i$ sufficiently large},
	\end{equation}
	for some constant $\rho\in(0,1)$.
\end{theorem}

\begin{proof}
	The proof follows the standard sequential quadratic programming (SQP) local convergence argument. By \textbf{(A1)} the functions are sufficiently smooth.
	By \textbf{(A2)}, the Jacobian $\nabla H(\stackrel{\circ}{U})$ has full row rank; thus LICQ holds for the equality constraints and
	$\nabla H(U)$ remains full row rank for $U$ sufficiently close to $\stackrel{\circ}{U}$.
	Consequently, the linearized equality constraint model is regular, and the QP \eqref{eq:mpsp_qp_again} has a unique KKT solution
	because $R\succ0$ makes the objective strictly convex. Moreover, \eqref{eq:mpsp_qp_again} is precisely the SQP subproblem for \eqref{eq:nlp_terminal} at $U^{i-1}$, with Hessian $R$ and constraint Jacobian $\mathcal{G}^{\,i-1}$, which coincides with $\nabla H(U^{i-1})$ by construction from the variational recursion. Therefore, the sequence $\{U^i\}$ generated by solving \eqref{eq:mpsp_qp_again} exactly is an SQP sequence.
	Assumption \textbf{(A3)} provides the second-order sufficient condition required for local SQP convergence, and \textbf{(A4)}
	ensures that the unit step is taken asymptotically. Hence, by the Boggs--Tolle local SQP convergence theorem \cite{boggs1989strategy}, the iterates converge locally to $(\stackrel{\circ}{U},\stackrel{\circ}{\lambda})$
	and the convergence is at least $Q$-linear. This yields \eqref{eq:qlinear_terminal}.
\end{proof}

\noindent In the next section, we present the iterative Linear Quadratic Regulator (iLQR) method for simple mechanical systems evolving on a Lie group, formulated using the same left‑trivialized state representation and intrinsic linearization tools as in our MPSP development. This intrinsic iLQR formulation enables a fair, structure‑preserving comparison between the two approaches, since both operate directly on the Lie‑group configuration space rather than in an ad hoc Euclidean embedding. By outlining the iLQR dynamics linearization, Riccati‑based backward pass, and nonlinear forward rollout on the same class of systems, we can quantitatively contrast its convergence behavior, computational complexity, and robustness properties against the proposed Lie‑group MPSP algorithm.
\subsection*{Comparison with iLQR and DDP}
We present the iLQR algorithm for simple mechanical systems evolving on
Lie groups, formulated using the same left-trivialized deviation model
derived in Section~\ref{Sec:MPSP_on_Lie_groups}. This serves both as a methodological
comparison and as a baseline for the numerical experiments in
Section~\ref{Sec:Numerical_simulation}.
Retain the left-trivialized deviation state
$\partial X_k = [\eta_k^\top,\; \partial v_k^\top]^\top \in \mathbb{R}^{2n}$
and control $u_k \in \mathbb{R}^n$ from Section~\ref{Sec:MPSP_on_Lie_groups}. The
nonlinear dynamics are propagated as
\begin{equation}\label{eq:nonlinear_dyn_ilqr}
	X_{k+1} = F(X_k, u_k),
	\qquad k = 1, \ldots, N-1,
\end{equation}
where $F$ denotes the discrete-time Lie-group integrator. The iLQR cost
is chosen as
\begin{equation}\label{eq:ilqr_cost}
	J = \frac{1}{2}\,\partial X_N^\top Q_N\,\partial X_N
	+ \sum_{k=1}^{N-1} \frac{1}{2}\,u_k^\top R_k\,u_k,
\end{equation}
where $Q_N \in \mathbb{R}^{2n \times 2n}$, $Q_N \succeq 0$ is the terminal
state weight, $R_k \in \mathbb{R}^{n \times n}$, $R_k \succ 0$ is the
running control weight, and $\partial X_N$ is the left-trivialized terminal
deviation computed from the nonlinear rollout. The running cost penalizes
only control effort, consistent with the MPSP formulation, while the
terminal cost provides a soft enforcement of the endpoint condition in
place of MPSP's hard constraint. At the $i$-th iteration, let $\{u_k^0\}_{k=1}^{N-1}$ and
$\{X_k^0\}_{k=1}^{N}$ denote the nominal control and state sequences
obtained by forward simulation of~\eqref{eq:nonlinear_dyn_ilqr}.
The left-trivialized linearization about this nominal trajectory yields
the time-varying linear system
\begin{equation}\label{eq:ilqr_lin}
	\partial X_{k+1} = A_k\,\partial X_k + B_k\,\partial u_k,
\end{equation}
where $A_k$ and $B_k$ are the intrinsic sensitivity matrices
from~\eqref{Eq:A_k_B_k_def}, reproduced here for clarity:
\begin{equation}\label{eq:AkBk_ilqr}
	A_k =
	\begin{bmatrix}
		A_{\eta}(\eta_{0,k},\, v_{0,k}) & h I_n \\
		0_{n \times n} & I_n + h\,A_v(v_{0,k})
	\end{bmatrix},
	\qquad
	B_k =
	\begin{bmatrix}
		0_{n \times n} \\
		h\,\mathbb{I}^{\sharp}
	\end{bmatrix}.
\end{equation}
Define the quadratic value function
$V_k(\partial X_k) \approx \frac{1}{2}\partial X_k^\top P_k\,\partial X_k
+ p_k^\top \partial X_k$,
initialized at the terminal step as
\begin{equation}\label{eq:ilqr_terminal}
	P_N = Q_N,
	\qquad
	p_N = Q_N\,\partial X_N^0,
\end{equation}
where $\partial X_N^0$ is the terminal deviation of the current nominal rollout.
The action-value (Q-function) expansion at step $k$ is
\begin{equation}
    \begin{split}
        \mathbf{Q}_{xx,k} &= A_k^\top P_{k+1}\,A_k,~\mathbf{Q}_{uu,k} = R_k + B_k^\top P_{k+1}\,B_k,~\mathbf{Q}_{ux,k} = B_k^\top P_{k+1}\,A_k,~\mathbf{Q}_{x,k}  = A_k^\top p_{k+1},~\mathbf{Q}_{u,k}  = R_k\,u_k^0 + B_k^\top p_{k+1}.
    \end{split}
\end{equation}
Minimizing the Q-function over $\partial u_k$ yields the optimal
affine feedback policy
\begin{equation}\label{eq:ilqr_policy}
	\partial u_k^* = K_k\,\partial X_k + d_k,
\end{equation}
where the feedback gain and feedforward term are
\begin{equation}\label{eq:KkDk}
	K_k = -\mathbf{Q}_{uu,k}^{-1}\,\mathbf{Q}_{ux,k}
	= -\bigl(R_k + B_k^\top P_{k+1}B_k\bigr)^{-1}
	B_k^\top P_{k+1}A_k,
\end{equation}
\begin{equation}\label{eq:dk}
	d_k = -\mathbf{Q}_{uu,k}^{-1}\,\mathbf{Q}_{u,k}
	= -\bigl(R_k + B_k^\top P_{k+1}B_k\bigr)^{-1}
	\bigl(R_k\,u_k^0 + B_k^\top p_{k+1}\bigr).
\end{equation}
Substituting \eqref{eq:ilqr_policy} back into the Q-function gives
the Riccati recursion for the value function matrices:
\begin{align}
	P_k &= \mathbf{Q}_{xx,k}
	- \mathbf{Q}_{xu,k}\,\mathbf{Q}_{uu,k}^{-1}\,\mathbf{Q}_{ux,k}= A_k^\top P_{k+1}\bigl(A_k + B_k K_k\bigr),
	\label{eq:Pk_riccati}\\
	p_k &= \mathbf{Q}_{x,k}-\mathbf{Q}_{xu,k}\,\mathbf{Q}_{uu,k}^{-1}\,\mathbf{Q}_{u,k}
	= \bigl(A_k + B_k K_k\bigr)^\top p_{k+1}
	+ K_k^\top R_k\,u_k^0.
	\label{eq:pk_riccati}
\end{align}
The recursion \eqref{eq:Pk_riccati} is the discrete-time Riccati
equation for the time-varying system $(A_k,B_k)$, driven by the
inertia-weighted control penalty $R_k$. Note that $\mathbf{Q}_{uu,k}
\succ 0$ for all $k$ since $R_k \succ 0$ and $P_{k+1} \succeq 0$,
so the gain matrices $K_k$ and $d_k$ are well-defined at every step. After completing the backward sweep from $k = N-1$ down to $k = 1$,
the control sequence is updated by rolling out the full nonlinear
dynamics~\eqref{eq:nonlinear_dyn_ilqr} with a line-search parameter
$\alpha \in (0, 1]$. Starting from the fixed initial condition
$X_1^{\mathrm{new}} = X_1^0$, at each step the updated control is
\begin{equation}\label{eq:ilqr_update}
	u_k^{\mathrm{new}}
	= u_k^0
	+ K_k\,\partial X_k^{\mathrm{new}}
	+ \alpha\,d_k,
	\qquad k = 1, \ldots, N-1,
\end{equation}
where $\partial X_k^{\mathrm{new}}
= \partial X_k^{\mathrm{new}} - X_k^0$ is the deviation of the new
rollout from the nominal trajectory. The Armijo line-search condition
\begin{equation}\label{eq:armijo}
	\Delta J(\alpha)
	\;\geq\;
	c_1 \alpha \sum_{k=1}^{N-1} d_k^\top \mathbf{Q}_{uu,k}\,d_k,
	\qquad c_1 \in (0,1),
\end{equation}
is used to select $\alpha$, ensuring sufficient decrease in the cost at
each iteration. The expected cost reduction is computed during the
backward pass as $\Delta J(\alpha)
= -\alpha\sum_k d_k^\top \mathbf{Q}_{u,k}
- \tfrac{1}{2}\alpha^2 \sum_k d_k^\top\mathbf{Q}_{uu,k}\,d_k$,
which avoids the need for an additional forward simulation to evaluate the acceptance condition.

Trajectory–based optimal control methods such as iterative LQR (iLQR) and Differential Dynamic Programming (DDP) 
share with MPSP the common feature of forward propagation of the full nonlinear dynamics. 
However, the manner in which sensitivity information is used, and the resulting optimization structure, 
differ fundamentally across these approaches. 
In iLQR and DDP, the dynamics are linearized (iLQR) or expanded to second order (DDP) along the nominal trajectory, 
and the running cost is quadratically approximated. 
This leads to a dynamic programming recursion in the form of a backward Riccati sweep, 
which yields both feedforward and feedback control terms. 
Consequently, iLQR and DDP compute locally optimal control updates by solving a sequence of 
quadratic subproblems coupled through the linearized dynamics.

In contrast, the proposed Lie–group MPSP formulation avoids Riccati equations and costate propagation entirely. 
The key structural simplification is that the terminal state is enforced as a hard constraint, 
while the cost penalizes only the control effort. 
As a result, only the terminal sensitivity map is required, and the full horizon dynamics 
do not appear in the optimization constraints. 
Each MPSP iteration therefore reduces to solving a static quadratic program with a closed–form solution, 
rather than a dynamic programming recursion. 
This trades exact second–order optimality for substantial computational savings, 
making MPSP particularly attractive for aggressive finite–time maneuvers on Lie groups, 
where solving Riccati recursions or TPBVPs can be prohibitively expensive. In the next section, we discuss the application of MPSP algorithm for the attitude maneuver for variable pitch quadrotor and single main rotor helicopter. 

\section{Application: Attitude Maneuver of Variable Pitch Quadrotor and Single Main Rotor Helicopter}\label{Sec:VPQ_SMRH}
In this section, we discuss the application of the MPSP algorithm to the finite‑time optimal attitude maneuver problem for a Variable Pitch Quadrotor (VPQ) and a Single Main Rotor Helicopter (SMRH). The VPQ configuration evolves on the nonlinear manifold $SO(3) \times \mathbb{R}^3$, where $SO(3)$ represents the attitude and $\mathbb{R}^3$ the body angular velocity, while the SMRH configuration space is $SO(3) \times \mathbb{R}^3 \times \mathbb{R}^3$, incorporating both the fuselage attitude/velocity and additional rotor or moment states. These configuration spaces are not Euclidean and the associated dynamics are highly nonlinear, which makes conventional optimal control techniques—typically formulated in $\mathbb{R}^n$ and relying on standard linearizations—difficult to apply directly or numerically fragile when embedded coordinates are used. The proposed Lie‑group MPSP framework addresses this challenge by working intrinsically on the underlying manifolds, using left‑trivialized variations and structure‑preserving linearization to compute optimal maneuvers without resorting to local coordinate charts or solving nonlinear two‑point boundary value problems.

\subsection{Variable pitch quadrotor}
\begin{figure}
	\centering
	\begin{tikzpicture}[scale=0.4]
		\draw[->,ultra thick,blue](3,4)node[below left]{$O_G$}--+(3,0)node[below right]{$e_2$(East)};
		\draw[->,ultra thick,blue](3,4)--+(0,-3)node[below]{$e_3$(Down)};
		\draw[->,ultra thick,blue](3,4)--+(-1.75,1.75)node[above right]{$e_1$(North)};
		\node[align=center,above right,blue,ultra thick]at(3,4){Ground Frame$(\mathcal{F}^G)$};
		\draw[fill=blue](3,4)circle(0.15);
		
		\draw[->,ultra thick,domain=11:15.5,brown]plot(\x,\x-1.5)node[above left]{$b_1$};
		\draw[->,ultra thick,domain=11:17,brown]plot(\x,0.125*\x+8.1)node[below right]{$b_2$};
		\draw[->,ultra thick,brown,rotate around={14:(11,9.5)}](11,9.5)--(11,5)node[right]{$b_3$};
		\node[align=center,below,brown,ultra thick]at(15,7){Body Frame$(\mathcal{F}^B)$};
		
		\draw[line width=1](5,7)--(13,8)--(17,12)--(9,11)--cycle;
		\draw[line width=1](5,7)--(17,12);
		\draw[line width=1](13,8)--(9,11);
		\path[fill=gray,opacity=0.3](5,7)--(13,8)--(17,12)--(9,11)--cycle;

		\draw[line width=2](5,7)--(4.75,8);
		\draw[line width=2](13,8)--(12.75,9);
		\draw[line width=2](17,12)--(16.75,13);
		\draw[line width=2](9,11)--(8.75,12);
		
		\draw[ultra thick,rotate around={10:(4.75,8)}](4.75,8)ellipse(1 and 0.5);
		\path[fill=gray,opacity=0.5,rotate around={10:(4.75,8)}](4.75,8)ellipse(1 and 0.5);
		\draw[->,thick,red,rotate around={10:(4.75,8)}](5.35,8)to[out=270,in=270](4.2,8);
		\draw[ultra thick,rotate around={10:(12.75,9)}](12.75,9)ellipse(1 and 0.5);
		\path[fill=gray,opacity=0.5,rotate around={10:(12.75,9)}](12.75,9)ellipse(1 and 0.5);
		\draw[->,thick,red,rotate around={10:(12.75,9)}](12,9.1)to[out=270,in=270](13.35,9.1);
		\draw[ultra thick,rotate around={10:(16.75,13)}](16.75,13)ellipse(1 and 0.5);
		\path[fill=gray,opacity=0.5,rotate around={10:(16.75,13)}](16.75,13)ellipse(1 and 0.5);
		\draw[->,thick,red,rotate around={10:(16.75,13)}](17.45,13.1)to[out=270,in=270](16.1,13.1);
		\draw[ultra thick,rotate around={10:(8.75,12)}](8.75,12)ellipse(1 and 0.5);
		\path[fill=gray,opacity=0.5,rotate around={10:(8.75,12)}](8.75,12)ellipse(1 and 0.5);
		\draw[->,thick,red,rotate around={10:(8.75,12)}](8.1,12.1)to[out=270,in=270](9.35,12.1);
		
		\draw[->,line width=2,rotate around={13:(4.75,8)}](4.75,8.2)--(4.75,9.5)node[left]{$F_{T_4}$};
		\draw[->,line width=2,rotate around={14:(12.7,9.2)}](12.7,9.2)--(12.75,10.5)node[right]{$F_{T_3}$};
		\draw[->,line width=2,rotate around={14:(16.7,13.2)}](16.7,13.2)--(16.7,14.5)node[right]{$F_{T_2}$};
		\draw[->,line width=2,rotate around={14:(8.75,12.2)}](8.7,12.2)--(8.7,13.5)node[left]{$F_{T_1}$};
		
		\draw[->,dashed,ultra thick](11,9.5)--(11,5.5)node[left]{$mg$};
		\draw[fill=brown](11,9.5)circle(0.15);
		\node[left,brown]at(10.8,9.6){$O_B$};
		
	\end{tikzpicture}
	\caption{Frames of references}
	\label{fig:Frames_VPQ}
\end{figure}
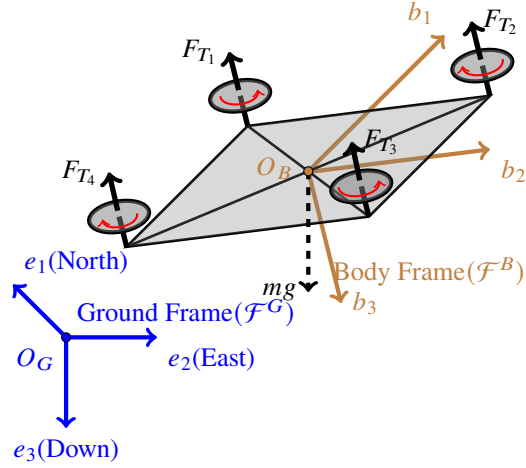

First we discuss the attitude dynamics of a variable pitch quadrotor. Consider an inertial frame $\left(\mathcal{O}_G,e_1,e_2,e_3\right)$ and body frame $\left(\mathcal{O}_B,b_1,b_2,b_3\right)$ (see Fig. \ref{fig:Frames_VPQ}). The attitude of a variable pitch quadrotor is represented by a rotation matrix, which is an element of the special orthogonal group $\mathrm{SO}\left(3\right)$, defined as $\mathrm{SO}\left(3\right) \triangleq \left\lbrace R \in \mathbb{R}^{3 \times 3}|~RR^T=\mathbf{I}_3,~\text{det}\left(R\right)=1 \right\rbrace$. The rotation matrix $R \in \mathrm{SO}\left(3\right)$ represents the attitude of the VPQ. The body angular velocity of the VPQ is denoted by $\omega:\mathbb{R}\to \mathbb{R}^3$ given by $\mathfrak{so}\left(3\right) \ni \hat{\omega}\left(t\right) = R^\top\left(t\right)\dot{R}\left(t\right)$. The hat map $\hat{\cdot}:\mathbb{R}^3 \to \mathfrak{so}\left(3\right)$ defined as
$\hat{x}y=x \times y,~\forall x,y \in \mathbb{R}^3$, where `$\times$' is the cross product in $\mathbb{R}^3$. In order to derive the equations of motion from \eqref{Eq:Dynamics_SMS_on_Lie_Group}, we have the matrix representation of $\operatorname{ad}_{\omega_i}\in L\left(\mathbb{R}^3;\mathbb{R}^3\right)$ as $\left[\operatorname{ad}_{\omega_i}\right]=\hat{\omega}_i$. The dual-adjoint operator $\operatorname{ad}^{*}_{\omega_i}\in L\left(\left(\mathbb{R}^3\right)^{*};\left(\mathbb{R}^3\right)^{*}\right)$ admits the matrix representation $\left[\operatorname{ad}^{*}_{\omega_i}\right]=-\left[\operatorname{ad}_{\omega_i}\right]=-\hat{\omega}_i$. The Lie algebra of $\mathrm{SO}\left(3\right)$ is $\mathfrak{so}\left(3\right)$, i.e, tangent space at identity element. The tangent bundle $\mathrm{TSO}\left(3\right)$ is isomorphic to $\mathrm{SO}\left(3\right) \times \mathrm{T}_{\mathbf{I}_3}\mathrm{SO}\left(3\right)$ or in other words $\mathrm{TSO}\left(3\right) \simeq \mathrm{SO}\left(3\right) \times \mathfrak{so}\left(3\right)$. Finally, we have the equations of motion of the VPQ as
\begin{equation}\label{Eq:VPQ_dyn}
	\begin{split}
		\dot{R}\left(t\right) &= R\left(t\right)\hat{\omega}\left(t\right)\\
		\dot{\omega}\left(t\right)+J^{-1}\left(\hat{\omega}\left(t\right)J\omega\left(t\right)\right) &= J^{-1}u\left(t\right)
	\end{split}
\end{equation}
The term $J \in \mathbb{R}^{3 \times 3}$ represents the inertia matrix of the VPQ and the control input $u \in \mathbb{R}^3 \simeq \mathfrak{so}^{*}\left(3\right)$. Let the components of control input $u=\left[l,m,n\right]^\top$.  The thrust from individual rotors is varied by changing their collective pitch input. It is assumed that all rotors are operating at the same nominal speed and thrust actuation is achieved by changing their collective pitch input. Blade element theory along with momentum theory \cite{leishman2006principles} is used to calculate thrust and torque of each rotor as a function of thrust coefficient $C_{T_i}~~\forall~i=1,2,3,4$. The thrust and moment equilibrium equations for the H configuration (similar to "X" configuration) for the normal flight mode is given below
\begin{equation}\label{Eq:Thrust_Moment_Equation}
	\begin{split}
		T &= \gamma K\left(C_{T1}+C_{T2}+C_{T3}+C_{T4}\right)\\
		l &= \gamma Kd\left(C_{T1}-C_{T2}-C_{T3}+C_{T4}\right)\\
		m &= \gamma Kd\left(C_{T1}+C_{T2}-C_{T3}-C_{T4}\right)\\
		n &= \gamma \frac{Kr}{\sqrt{2}}\left(C_{T1}^{\frac{3}{2}}-C_{T2}^{\frac{3}{2}}+C_{T3}^{\frac{3}{2}}-C_{T4}^{\frac{3}{2}}\right)
	\end{split}
\end{equation}
where $K=\rho\pi r^{2}V^2_{\text{tip}}$, $\rho$ is the density of air, $r$ is the rotor radius, $d$ is the distance between a rotor and center of mass, $V_{\text{tip}}=\Omega r$, and $\Omega$ is the angular speed of the rotor and $\gamma$ is chosen as $+1$ for normal flight mode and $-1$ for inverted flight mode. The term $\gamma$ is introduced to ensure that thrust and moment equations properly captures the yaw dynamics while in the normal and inverted flight modes. For more details regarding rotor dynamics, one can refer to \cite{leishman2006principles,chipade2018advanced}. The optimal attitude maneuver from the initial condition $\left(R\left(0\right),\omega \left(0\right)\right)$ to the desired final state $\left(R\left(t_f\right),\omega \left(t_f\right)\right)$ while minimizing the control input $u$. The variation of a rotation matrix $R$ is represented using the exponential map $\operatorname{exp}: \mathfrak{so}\left(3\right) \to \mathrm{SO}\left(3\right)$ is defined as $R^{\epsilon}=R\operatorname{exp}\left(\epsilon \hat{\underline{\eta}}\right)$, where $\epsilon \in \left(-c,c\right)$ for $c \in \mathbb{R}_{>0},~\underline{\eta} \in \mathbb{R}^{3}$. The infinitesimal variation of the rotation matrix is found as
\begin{equation}
	\partial R =\left. \frac{\operatorname{d}}{\operatorname{dt}}\right\vert_{\epsilon=0} \left(R \operatorname{exp}\left(\epsilon \hat{\underline{\eta}}\right)\right)=R \hat{\underline{\eta}}
\end{equation}
Let the variation in angular velocity $\omega$ be $\partial \omega$. Now define the combined state vector $\partial X =\left[\underline{\eta}^\top,\partial \omega^\top\right]$. Now the discrete time evolution of the state $\partial X$ can be derived as
\begin{equation}
	\partial X_{k+1}=A_kX_k+B_k \partial u_k
\end{equation}
where
\begin{equation*}
	\begin{split}
		A_k &= \begin{bmatrix}
			\mathbf{I}_3-h\left(\hat{\omega}_{0k}+\hat{\partial \omega_k}\right) & h \mathbf{I}_3 \\
			\mathbf{0}_{3 \times 3} & \mathbf{I}_3+h\left(-J^{-1}\hat{\omega}_{0k}J+\hat{\omega}_{0k}\right)
		\end{bmatrix},~B_k=\begin{bmatrix}
		\mathbf{0}_{3 \times 3}\\ h J^{-1}
		\end{bmatrix}
	\end{split}
\end{equation*}
Since the end points are fixed, we have
\begin{equation}
	\partial X_N =\sum_{k=1}^{N-1}G_k \partial \underline{u}_k
\end{equation}
Let $\left\lbrace u^{i-1}_k \right\rbrace^{N-1}_{k=1}$ be the optimal control sequence found using MPSP algorithm at $\left(i-1\right)^{\operatorname{th}}$ iteration. Now we need to find the optimal control sequence $\left\lbrace u^{i}_k \right\rbrace^{N-1}_{k=1}$ for th current $i^{\operatorname{th}}$ iteration. Now the objective of MPSP is stated as below
\begin{equation}
	\begin{aligned}
		& \operatorname{minimize}_{\left\{\partial u_k\right\}_{k=1}^{N-1}} \quad J:=\sum_{k=1}^{N-1} \frac{1}{2} \left(u^{i-1}_k-\partial u^{i}_k\right)^\top R_k \left(u^{i-1}_k-\partial u^{i}_k\right) \\
		& \text { subject to }\left\{\begin{array}{l}
			\partial X^{i-1}_N =\sum_{k=1}^{N-1}G_k\partial u^{i}_k
		\end{array}\right.
	\end{aligned}
\end{equation}
Now from Theorem \ref{Th:MPSP_1}, we find the optimal control which solves the above optimal control problem.

\subsection{Single main rotor helicopter}
\begin{figure}[h!]
	\centering
	\includegraphics[width= 8cm, height = 3 cm]{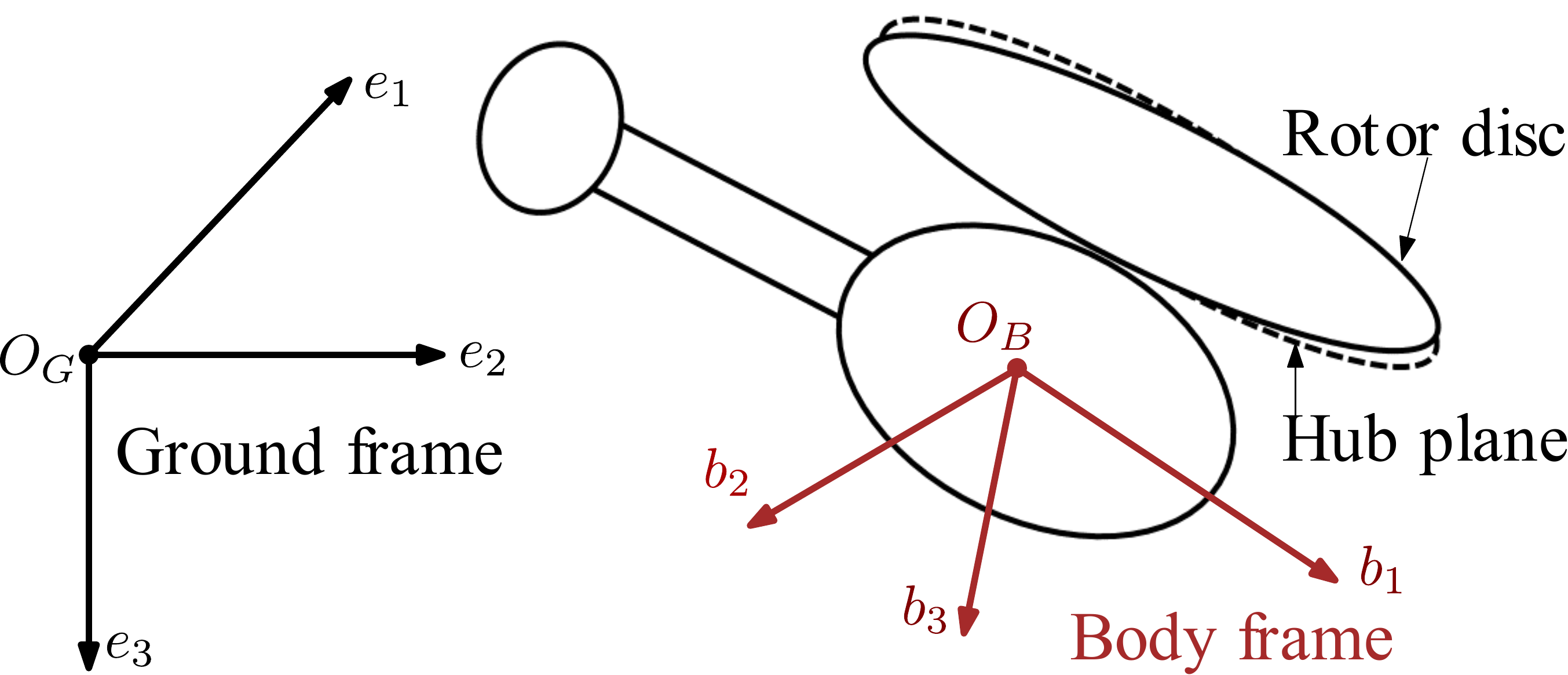}
	\caption{Frames}
	\label{fig:Frames_SMRH}
\end{figure}

In contrast to a VPQ, a single rigid body assumption is not applicable to a SMRH. Because of the presence of rotor and fuselage interaction, rotor and fuselage must be modeled as separate rigid bodies (see Fig. \ref{fig:Frames_SMRH}).  Unlike VPQ, where a single rigid-body assumption suffices and STA-based robust control has been successfully demonstrated \cite{krishna2022robust}, the SMRH requires explicit modeling of rotor–fuselage coupling. This motivates the extension of MPSP to handle actuator dynamics and parametric uncertainty. The rotor is modeled as a rigid disk with the flapping motion represented by two sets of first-order differential equations. Let the components of angular velocity of fuselage be $\omega=\left[p,q,r\right]^\top$. Next we discuss the flapping motion of the main rotor. Let $\beta_{1c} \in \mathbb{R}$ and $\beta_{1s} \in \mathbb{R}$ represent the longitudinal and lateral tilt of the main rotor measured from the hub plane. The control inputs be $\left[\theta_{\operatorname{lon}},\theta_{\operatorname{lat}}\right]^\top\in \mathbb{R}^2$ which represent the lateral and longitudinal cyclic blade pitch angles respectively. Then the flapping motion of the main rotor is described by
\begin{equation}\label{Eq:Flap_eqns}
	\begin{split}
		\dot{\beta}_{1c}&=-\frac{1}{\tau_m}\beta_{1c}-\frac{k_\beta}{2\Omega_{\text{mr}} I_\beta}\beta_{1s}-q-\frac{1}{\Omega_{\text{mr}}\tau_m}p+\frac{1}{\tau_m}\theta_{\operatorname{lon}}\\
		\dot{\beta}_{1s}&=-\frac{1}{\tau_m}\beta_{1s}+\frac{k_\beta}{2\Omega_{\text{mr}} I_\beta}\beta_{1c}-p+\frac{1}{\Omega_{\text{mr}}\tau_m}q+\frac{1}{\tau_m}\theta_{\operatorname{lat}}
	\end{split}
\end{equation}  

The parameters appearing in the above equations are main rotor time constant $\tau_m \in \mathbb{R}_{>0}$, main rotor angular speed $\Omega_{\text{mr}}\in \mathbb{R}_{>0}$, blade root stiffness $k_\beta \in \mathbb{R}_{>0}$ and the blade moment of inertia about the flap hinge $I_\beta \in \mathbb{R}_{>0}$. The net external rolling and pitching moments acting on the fuselage are caused due to the tilting of the main rotor thrust $T_{\operatorname{mr}}$ and the restoring moments acting on the flap hinge due to $k_{\beta}$. Hence the moments $l,m$ are expressed as $l=\left(k_{\beta}+h_{\operatorname{mr}}T_{\operatorname{mr}}\right)\beta_{1s},~m=\left(k_{\beta}+h_{\operatorname{mr}}T_{\operatorname{mr}}\right)\beta_{1c}$. The yawing moment on the fuselage $n$, applied by the tail rotor, is modeled by using a first-order differential equation as follows
\begin{equation}\label{Eq:n_tr_diff_eqn}
	\dot{n}_{\text{tr}}=-\frac{n_{\text{tr}}}{\tau_t}-K_tr+\frac{K_t}{\tau_t}\theta_{\operatorname{tail}}
\end{equation}
where $\tau_t, K_t \in \mathbb{R}_{>0}$ denote the time constant and control effectiveness of the tail rotor respectively. Let the control input be $\mathbb{R}^3 \ni  u= \left[ \theta_{\operatorname{lon}},~ \theta_{\operatorname{lat}},~ \theta_{\operatorname{tail}}\right]^\top$. From the rotor dynamics and moment expressions, we get the main rotor dynamics as
\begin{equation}\label{Eq:Moment_dynamics}
	\dot{M}=AM-K\omega+Bu
\end{equation}
where
\begin{equation*}
A \triangleq \begin{bmatrix}
	-\frac{1}{\tau_m} ~& \frac{k_\beta}{2\Omega I_\beta} & 0\\
	-\frac{k_\beta}{2\Omega I_\beta} ~& -\frac{1}{\tau_m} & 0\\
	0 & 0 & -\frac{1}{\tau_t}
\end{bmatrix},~K \triangleq \begin{bmatrix}
\left(hT+k_\beta\right)~ & -\frac{\left(hT+k_\beta\right)}{\Omega\tau_m} & 0\\
\frac{\left(hT+k_\beta\right)}{\Omega\tau_m}~ & \left(hT+k_\beta\right) & 0\\
0 & 0 & K_t
\end{bmatrix},~B \triangleq\text{diag}\left \lbrace \frac{\left(hT+k_\beta\right)}{\tau_m},~\frac{\left(hT+k_\beta\right)}{\tau_m},~\frac{K_t}{\tau_t}\right \rbrace
\end{equation*}

Compared to VPQ, the SMRH model introduces additional complexity because the rotor dynamics are coupled with fuselage motion through aerodynamic damping and flap stiffness. This coupling prevents the system from being treated as a simple mechanical system (SMS) on a Lie group. Prior robust control designs such as the Super Twisting Algorithm (STA) \cite{krishna2021super} have shown finite-time convergence and disturbance rejection for SMRH, but those approaches did not address optimality. Here, we extend the MPSP framework to explicitly handle the rotor–fuselage interaction while minimizing control effort.

From the above discussion, it can be concluded that the SMRH mathematical model does not have the structure of a SMS. However, following the procedure discussed in Sec \ref{Sec:MPSP_on_Lie_groups}, MPSP algorithm can be implemented which is discussed next. Let the variation in moment $M$ be $\partial M$. Now, we define the state vector $\partial X=\left[\eta^\top,\partial \omega^\top , \partial M^\top\right]^\top$. The discrete-time evolution of the state vector $\partial X$ can be derived as
\begin{equation*}
	\partial X_{k+1}=A_k \partial X_k+B_k \partial u_k
\end{equation*}
where
\begin{equation*}
	A_k \triangleq \begin{bmatrix}
		\mathbf{I}_3-h\left(\hat{\omega}_{0k}+\hat{\partial \omega_k}\right) & h \mathbf{I}_3 & \mathbf{0}_{3 \times 3}\\
		\mathbf{0}_{3 \times 3} & \mathbf{I}_3+h \left(-J^{-1}\hat{\omega}_{0k}J+\hat{\omega}_{0k}\right) & h J^{-1}\\
		\mathbf{0}_{3 \times 3} & \mathbf{I}_3-hK & \mathbf{I}_3+hA
	\end{bmatrix},~B_k=\begin{bmatrix}
	\mathbf{0}_{3 \times 3}\\
	\mathbf{0}_{3 \times 3}\\
	h B
	\end{bmatrix}
\end{equation*}

Now, the MPSP algorithm can be implemented using Theorem \ref{Th:MPSP_1}. This formulation highlights the novelty of applying MPSP to SMRH: despite the absence of SMS structure, the algorithm remains tractable and yields closed-form updates. In order to provide a principled benchmark for the discrete‑time Lie‑group MPSP algorithm on both VPQ and SMRH, we next derive the continuous‑time necessary conditions for finite‑time optimal attitude control directly on $\mathrm{SO}(3)$. Working at the level of the underlying rotation group allows us to formulate the optimal control problem in a coordinate‑free manner, using intrinsic rigid‑body dynamics and variational arguments rather than any particular choice of Euler angles or parameterization. The resulting Pontryagin‑type conditions yield a two‑point boundary value problem (TPBVP) on $\mathrm{SO}(3)$ that characterizes continuous‑time optimal attitude maneuvers for each vehicle. These continuous‑time TPBVP solutions serve as high‑fidelity references against which the discrete‑time MPSP trajectories can be compared in the numerical study, thereby quantifying the accuracy and near‑optimality of the proposed Lie‑group MPSP framework for both the VPQ and the SMRH.
\section{Continuous-Time Attitude Control of Variable Pitch Quadrotor and Single Main Rotor Helicopter }\label{Sec:VPQ_SMRH_CT} 
In this section, we solve the finite-time optimal attitude maneuver of VPQ and SMRH using TPBVP.
\subsection{Variable Pitch Quadrotor}

The optimal attitude maneuver objective is to transfer from an initial attitude $R_0 \in \mathrm{SO}\left(3\right)$ and angular velocity $\omega_0 \in \mathbb{R}^{3 \times 1}$ to desired values $\left(R_{F},\omega_{F}\right) \in \mathrm{SO}\left(3\right)$ and angular velocity $\omega_0 \in \mathbb{R}^{3 \times 1}$ to desired values $\left(R_{F},\omega_{F}\right) \in  \mathrm{SO}\left(3\right) \times \mathbb{R}^{3 \times 1}$ while minimizing the objective function
\begin{equation}\label{Eq:VPQ_CT_Obj}
	\stackrel{\text{minimize}}{u} \mathcal{J}=\frac{1}{2} \int_{0}^{t_f} \left(u^\top Q u\right)\operatorname{d}t
\end{equation}
such that $\left(R\left(t_f\right),\omega\left(t_f\right)\right)=\left(R_F,\omega_F\right)$ subject to the dynamics \eqref{Eq:VPQ_dyn}. In order to find the necessary conditions of optimality, we need to find the infinitesimal variation of $R^\top \dot{R}$, which can be derived as
	\begin{equation*}
	\begin{split}
		\partial \left(R^\top \dot{R}\right) &= \partial R \dot{R}+R^\top \partial \dot{R}=-\hat{\eta}R^\top \dot{R}+R^\top \left(\dot{R}\hat{\eta}+R \hat{\dot{\eta}}\right)=-\hat{\eta} \hat{\omega}+\hat{\omega}\hat{\eta}+\hat{\dot{\eta}}=\widehat{\left(\dot{\eta}+\hat{\omega}\eta\right)}
	\end{split}
\end{equation*}
The next theorem establishes the necessary conditions for the optimality for achieving the problem objective given in \eqref{Eq:VPQ_CT_Obj}.
\begin{theorem}
	Let $\stackrel{\circ}{u }\left(t\right),~\forall t \in \left[0,t_f\right]$ be the optimal control that solves the problem objective \eqref{Eq:VPQ_CT_Obj}, then the necessary conditions for optimality are
	\begin{itemize}
		\item Optimal control
		\begin{equation}\label{Eq:VPQ_CT_Opt_Cont}
			\begin{split}
				\stackrel{\circ}{u }\left(t\right)&=-Q^{-1}\stackrel{\circ}{\lambda}_{\omega}\left(t\right)
			\end{split}
		\end{equation}
		\item Multiplier equations
		\begin{equation}\label{Eq:VPQ_CT_ME}
			\begin{split}
				\frac{\operatorname{d}}{\operatorname{d}t} \left(\stackrel{\circ}{\lambda}_R \left(t\right)\right) &= -\widehat{\stackrel{\circ}{\omega}\left(t\right)} \stackrel{\circ}{\lambda}_R \left(t\right)\\
					\frac{\operatorname{d}}{\operatorname{d}t} \left(\stackrel{\circ}{\lambda}_\omega \left(t\right)\right) &= \left(J^{-1}\widehat{\left(J \stackrel{\circ}{\omega} \left(t\right)\right)}-\widehat{\stackrel{\circ}{\omega} \left(t\right)}\right)\stackrel{\circ}{\lambda}_\omega \left(t\right)
			\end{split}
		\end{equation}
		\item Boundary conditions
		\begin{equation}
			\left(\stackrel{\circ}{R}\left(0\right),\stackrel{\circ}{\omega} \left(0\right) \right)=\left(R_0,\omega_0\right),~~\left(\stackrel{\circ}{R}\left(t_f\right),\stackrel{\circ}{\omega} \left(t_f\right) \right)=\left(R_F,\omega_F\right)
		\end{equation}
	\end{itemize}
	where $\stackrel{\circ}{\lambda}_R \left(t\right) ,\stackrel{\circ}{\lambda}_{\omega}\left(t\right) \in \left(\mathbb{R}^{3 \times 1}\right)^2$ are the optimal Lagrange multipliers.
\end{theorem}
\begin{proof}
	Consider the augmented performance index,
	\begin{equation}
		\mathcal{J}_a =\int_{0}^{t_f} \left\lbrace\frac{1}{2} u^\top Q u +\lambda^\top_R\left(\hat{\omega}-R^\top \dot{R}\right)^\vee+\lambda^\top_\omega\left(u-\hat{\omega}J \omega-J \dot{\omega}\right) \right\rbrace \operatorname{d}t
	\end{equation}
	where $\lambda_R,\lambda_\omega  \in \left(\mathbb{R}^{3 \times 1}\right)^2$ are the Lagrange multipliers. The variation of the augmented performance index is given by
	\begin{equation*}
		\partial \mathcal{J}_a =\int_{0}^{t_f} \left\lbrace u^\top Q \partial u +\lambda^\top_R \left(\partial \omega -\partial \left(R^\top \dot{R}\right)^\vee\right)+\lambda^\top_{\omega} \left(\partial u -\hat{\omega} J\partial \omega -\widehat{\partial \omega}J \omega -J \partial \dot{\omega}\right)\right\rbrace\operatorname{d}t
	\end{equation*}
	Substituting the infinitesimal variation of $R^\top \dot{R}$, we get
	\begin{equation*}
		\partial \mathcal{J}_a =\int_{0}^{t_f} \left\lbrace u^\top Q \partial u +\lambda^\top_R \left(\partial \omega -\dot{\eta} -\hat{\omega} \eta\right)+\lambda^\top_{\omega} \left(\partial u -\hat{\omega} J\partial \omega -\widehat{\partial \omega}J \omega -J \partial \dot{\omega}\right)\right\rbrace\operatorname{d}t
	\end{equation*}
	By applying integration by parts,
	\begin{equation*}
		\partial \mathcal{J}_a =\int_{0}^{t_f} \left\lbrace u^\top Q \partial u +\lambda^\top_R \left(\partial \omega -\hat{\omega} \eta\right) +\dot{\lambda}^\top_R \eta +\lambda^\top_{\omega} \left(\partial u -\hat{\omega} J\partial \omega -\widehat{\partial \omega}J \omega\right)+\dot{\lambda}^\top_\omega J \partial \omega \right\rbrace\operatorname{d}t-\left.\left\lbrace \lambda^\top_R \eta +\lambda^\top_\omega J\delta\omega\right\rbrace\right \vert^{t_f}_0
	\end{equation*}
	Since the initial attitude and angular velocity are fixed, we have $\eta\left(0\right)=\mathbf{0}_{3 \times 1}$ and $\partial \omega \left(0\right)=\mathbf{0}_{3 \times 1}$ 
	\begin{equation*}
		\partial \mathcal{J}_a =\int_{0}^{t_f} \left\lbrace u^\top Q \partial u +\lambda^\top_R \left(\partial \omega -\hat{\omega} \eta\right) +\dot{\lambda}^\top_R \eta +\lambda^\top_{\omega} \left(\partial u -\hat{\omega} J\partial \omega -\widehat{\partial \omega}J \omega\right)+\dot{\lambda}^\top_\omega J \partial \omega \right\rbrace\operatorname{d}t
	\end{equation*}
	 Using the scalar triple product of three vectors $a,b,c \in \mathbb{R}^3$, $[a,b,c]=a^\top\left(b \times c\right)=b^\top \left(c \times a\right)=c^\top \left(a \times b\right)$. The variation of $\mathcal{J}_a$ can be simplified as
	\begin{equation*}
			\partial \mathcal{J}_a =\int_{0}^{t_f} \left\lbrace  \partial u^\top \left(Q u+\lambda_\omega\right)+\eta^\top \left(\hat{\omega}\lambda_R +\dot{\lambda}_R\right)+\partial \omega^\top \left(-\widehat{J \omega}-J\hat{\lambda}_\omega \omega+J \dot{\lambda}_\omega +\lambda_R\right)\right\rbrace\operatorname{d}t
	\end{equation*}
	Choosing multiplier equations such that expressions multiplied with the variations are identically zero, we get the necessary conditions given in \eqref{Eq:VPQ_CT_Opt_Cont} and \eqref{Eq:VPQ_CT_ME}.
\end{proof}

\noindent In the next section, we derive the continuous‑time necessary conditions for the finite‑time optimal attitude maneuver of the SMRH. Building on the geometric rigid‑body model and rotor–fuselage coupling dynamics introduced earlier, we formulate a continuous‑time optimal control problem directly on $\mathrm{SO}(3)\times \mathbb{R}^3 \times \mathbb{R}^3$, where the rotation matrix, angular velocity, and rotor moment evolve according to the SMRH equations of motion. Using variational calculus and Pontryagin’s maximum principle in this Lie‑group setting, we obtain the associated boundary‑value problem linking the state and costate trajectories, along with the optimal control law expressed in terms of the adjoint variables. These continuous‑time necessary conditions provide a rigorous characterization of optimal SMRH attitude maneuvers and furnish a high‑accuracy reference solution that is later used to evaluate the performance and near‑optimality of the discrete‑time MPSP algorithm.
\subsection{Single Main Rotor Helicopter}

The optimal attitude maneuver objective is to transfer from an initial conditions $\left(R_0,\omega_0,M_0\right) \in \mathrm{SO}\left(3\right) \times \left(\mathbb{R}^{3 \times 1}\right)^2$ to the final conditions $\left(R_F,\omega_F,M_F\right) \in \mathrm{SO}\left(3\right) \times \left(\mathbb{R}^{3 \times 1}\right)^2 $ while minimizing the objective function
\begin{equation}\label{Eq:SMRH_CT_Obj}
	\stackrel{\text{minimize}}{u} \mathcal{J}=\frac{1}{2} \int_{0}^{t_f} \left(u^\top Q u\right)\operatorname{d}t
\end{equation}
such that $\left(R\left(t_f\right),\omega\left(t_f\right),M\left(t_f\right)\right)=\left(R_F,\omega_F,M_F\right)$ subject to the dynamics \eqref{Eq:Moment_dynamics}. The next theorem establishes the necessary conditions for the optimality for achieving the problem objective given in \eqref{Eq:SMRH_CT_Obj}.
\begin{theorem}
	Let $\stackrel{\circ}{u }\left(t\right),~\forall t \in \left[0,t_f\right]$ be the optimal control that solves the problem objective \eqref{Eq:SMRH_CT_Obj}, then the necessary conditions for optimality are
	\begin{itemize}
		\item Optimal control
		\begin{equation}\label{Eq:SMRH_CT_Opt_Cont}
			\begin{split}
				\stackrel{\circ}{u }\left(t\right)&=-Q^{-1}B^\top \stackrel{\circ}{\lambda}_{M}\left(t\right)
			\end{split}
		\end{equation}
		\item Multiplier equations
		\begin{equation}\label{Eq:VSMRH_CT_ME}
			\begin{split}
				\frac{\operatorname{d}}{\operatorname{d}t} \left(\stackrel{\circ}{\lambda}_R \left(t\right)\right) &= -\widehat{\stackrel{\circ}{\omega}\left(t\right)} \stackrel{\circ}{\lambda}_R \left(t\right)\\
				\frac{\operatorname{d}}{\operatorname{d}t} \left(\stackrel{\circ}{\lambda}_\omega \left(t\right)\right) &= \left(J^{-1}\widehat{\left(J \stackrel{\circ}{\omega} \left(t\right)\right)}-\widehat{\stackrel{\circ}{\omega} \left(t\right)}\right)\stackrel{\circ}{\lambda}_\omega \left(t\right)-J^{-1} \stackrel{\circ}{\lambda}_R -J^{-1} K^\top \stackrel{\circ}{\lambda}_M\\
				\frac{\operatorname{d}}{\operatorname{d}t} \left(\stackrel{\circ}{\lambda}_M \left(t\right)\right) &=- \stackrel{\circ}{\lambda}_\omega -A^\top \stackrel{\circ}{\lambda}_M
			\end{split}
		\end{equation}
		\item Boundary conditions
		\begin{equation}
			\left(\stackrel{\circ}{R}\left(0\right),\stackrel{\circ}{\omega} \left(0\right) , \stackrel{\circ}{M}\left(0\right) \right)=\left(R_0,\omega_0,M_0\right),~~\left(\stackrel{\circ}{R}\left(t_f\right),\stackrel{\circ}{\omega} \left(t_f\right), \stackrel{\circ}{M}\left(t_f\right) \right)=\left(R_F,\omega_F,M_F\right)
		\end{equation}
	\end{itemize}
	where $\stackrel{\circ}{\lambda}_R \left(t\right) ,\stackrel{\circ}{\lambda}_{\omega}\left(t\right),\stackrel{\circ}{\lambda}_{M}\left(t\right) \in \left(\mathbb{R}^{3 \times 1} \right)^3$ are the optimal Lagrange multipliers.
\end{theorem}
\begin{proof}
	Consider the augmented performance index,
	\begin{equation}
		\mathcal{J}_a =\int_{0}^{t_f} \left\lbrace\frac{1}{2} u^\top Q u +\lambda^\top_R\left(\hat{\omega}-R^\top \dot{R}\right)^\vee+\lambda^\top_\omega\left(u-\hat{\omega}J \omega-J \dot{\omega}\right) +\lambda^\top_M \left(AM-K\omega+Bu-\dot{M}\right)\right\rbrace \operatorname{d}t
	\end{equation}
	where $\lambda_R,\lambda_\omega,\lambda_M  \in \left(\mathbb{R}^{3 \times 1}\right)^3$ are the Lagrange multipliers. The variation of the augmented performance index is given by
	\begin{equation*}
		\begin{split}
			\partial \mathcal{J}_a &=\int_{0}^{t_f} \left\lbrace u^\top Q \partial u +\lambda^\top_R \left(\partial \omega -\partial \left(R^\top \dot{R}\right)^\vee\right)+\lambda^\top_{\omega} \left(\partial u -\hat{\omega} J\partial \omega -\widehat{\partial \omega}J \omega -J \partial \dot{\omega}\right)+\lambda^\top_M \left(A\partial M-K\partial \omega +B \partial u-\partial \dot{M}\right)\right\rbrace\operatorname{d}t
		\end{split}
	\end{equation*}
	Applying integration by parts
	\begin{equation*}
		\begin{split}
			\partial \mathcal{J}_a &=\int_{0}^{t_f} \left\lbrace u^\top Q \partial u + \lambda^\top_R \left(\partial \omega -\hat{\omega} \eta\right) +\dot{\lambda}^\top_R \eta +\lambda^\top_{\omega} \left(\partial u -\hat{\omega} J\partial \omega -\widehat{\partial \omega}J \omega\right)+\dot{\lambda}^\top_\omega J \partial \omega  \right.\\
			&\qquad \left. +\lambda^\top_M \left(A\partial M-K\partial \omega+B \partial u-\partial \dot{M}\right)+\dot{\lambda}^\top_M \partial M  \right\rbrace\operatorname{d}t\\
			&=\int_{0}^{t_f} \left\lbrace \partial u^\top \left(Q u+B^\top \lambda_{M}\right)+\eta^\top \left(\hat{\omega} \lambda_R+\dot{\lambda}_R\right)+\partial \omega^\top \left(-\widehat{J \omega}-J\hat{\lambda}_\omega \omega +J \dot{\lambda}_\omega+\lambda_R-K^\top \lambda_M\right) \right.\\
			&\qquad \left. +\partial M^\top \left(A^\top \lambda_M+\lambda_\omega+\dot{\lambda}_M\right)\right\rbrace\operatorname{d}t
		\end{split}
	\end{equation*}
	By choosing the multiplier equations such that expressions multiplied with variations are identically zero, we get the necessary conditions given in \eqref{Eq:SMRH_CT_Opt_Cont} and \eqref{Eq:VSMRH_CT_ME}.
\end{proof}

\noindent In the next section, we numerically validate the proposed MPSP algorithm by simulating finite‑time optimal attitude flipping maneuvers for both the VPQ and the SMRH. Using the intrinsic Lie‑group formulation developed earlier, we generate discrete‑time trajectories on $\mathrm{SO}(3)$ (for VPQ) and on the extended configuration space including rotor dynamics (for SMRH), and apply MPSP to compute the corresponding control inputs that achieve aggressive 180° flip maneuvers within prescribed time horizons. The resulting state and control histories are then compared against benchmark solutions obtained from the continuous‑time TPBVP formulations and, where applicable, against trajectories produced by iLQR. This numerical study highlights the accuracy, convergence behavior, and computational efficiency of the Lie‑group MPSP scheme in realistic aerospace maneuvering scenarios.

\section{Numerical Simulation}\label{Sec:Numerical_simulation}

This section presents numerical simulations for the Variable Pitch Quadrotor (VPQ) and Single Main Rotor Helicopter (SMRH) flipping maneuvers to evaluate the proposed Lie-group MPSP algorithm. The results include maneuver trajectories, control histories, iteration-wise terminal deviation behavior, and robustness under randomized initialization. We additionally compare the MPSP-generated trajectories with continuous-time TPBVP solutions to assess accuracy. The numerical study highlights three distinct methods for solving aggressive finite-time attitude maneuvers on nonlinear configuration spaces such as SO(3): the proposed discrete-time Lie-group MPSP method, the continuous-time optimal control formulation based on necessary/sufficient conditions and TPBVP solution, and the trajectory-optimization framework iLQR. Although all three methods aim to compute control inputs that realize the same terminal maneuver, they differ fundamentally in how optimality, feasibility, and computational structure are handled. The following section shows that MPSP-generated trajectories closely match the continuous-time benchmark for both the variable-pitch quadrotor (VPQ) and the single-main-rotor helicopter (SMRH), with negligible deviation in state and control histories. In particular, for the SMRH case, despite the additional rotor-fuselage coupling and actuator-related complexity, the MPSP trajectory remains close to the TPBVP solution and converges in only a small number of iterations. All simulations are carried out using a fixed discrete time step of $h=1~\mathrm{ms}$.  
The maneuver horizon is chosen as VPQ: $T = 0.6~\mathrm{s}$ ($N=600$ steps), SMRH: $T = 1.2~\mathrm{s}$ ($N=1200$ steps).

\subsection{Variable Pitch Quadrotor (VPQ)}

The inertia matrix of the VPQ is chosen as $J=[J_{ij}]$, where $J_{xx}=0.0122,~J_{yy}=0.0266,~J_{zz}=0.0387,~J_{xy}=0.0003,~J_{xz}=0.00056,~J_{yz}=0.00032$. At each MPSP iteration, the terminal sensitivity matrices are computed using the intrinsic left-trivialized linearization on $SO(3)$. These matrices remain well-conditioned throughout the maneuver, confirming numerical stability of the Lie-group deviation model. The VPQ is commanded to perform a $180^{\circ}$ flip about the roll axis. The initial attitude and angular velocity are chosen as $[\phi(0),\theta(0),\psi(0)]^\top=[0,0,0]^\top,
[p(0),q(0),r(0)]^\top=[0,0,0]^\top$ and the desired terminal state at $t_f=0.6$~s is $[\phi(t_f),\theta(t_f),\psi(t_f)]^\top=[180^{\circ},0,0]^\top,
[p(t_f),q(t_f),r(t_f)]^\top=[0,0,0]^\top$. The VPQ results are shown in Figs.~\ref{fig:Euler_ang_VPQ}--\ref{fig:VPQ_ErrNorm_itr}. Since the VPQ attitude dynamics can be represented as a rigid-body model on SO(3), the continuous-time TPBVP solution serves as a strong benchmark for assessing the quality of the discrete-time approximation. The MPSP algorithm converges in six iterations for the nominal VPQ flip. The terminal deviation norm $\|\partial X_N\|$ exhibits a non-monotonic pattern (Fig.~\ref{fig:VPQ_ErrNorm_itr}), which is expected because MPSP performs full-step feasibility corrections based on first-order sensitivity information rather than a line-search or merit-function-based update. When the iterate lies outside the local linearization region, the measured terminal deviation may temporarily increase; once inside the contraction region predicted by Theorem~III.3, rapid reduction occurs. The MPSP trajectory (solid) closely matches the continuous-time TPBVP benchmark (dashed), demonstrating the accuracy of the Lie-group sensitivity model and yielding smooth, physically realizable control inputs. The close overlap between the MPSP and continuous-time responses suggests that the intrinsic linearized deviation model and the recursively computed terminal sensitivity matrices are sufficiently accurate for aggressive flip maneuvers. In other words, the discretization used by MPSP does not introduce any substantial deviation in the control input profile, and the final control history is nearly indistinguishable from the continuous-time optimal solution. From a numerical viewpoint, this is significant because MPSP avoids the repeated forward-backward shooting required by the TPBVP while still recovering essentially the same maneuver. Thus, for the VPQ, MPSP appears to provide an excellent balance between computational tractability and near-optimal performance. The continuous-time solution is the most faithful representation of the original optimal control problem, but it comes at the cost of solving a nonlinear two-point boundary value problem on a Lie group. Such shooting-based solvers are well known to be initialization-sensitive and computationally intensive. By contrast, the proposed MPSP update reduces each iteration to a static quadratic program with a closed-form update law based only on terminal sensitivity information. Thus, in the numerical section, the continuous-time method should be interpreted as a benchmark for optimality, whereas MPSP should be interpreted as a computationally efficient near-optimal solver. The simulations indicate that the loss in exact optimality is very small, while the gain in numerical simplicity is substantial. As already stated earlier in the manuscript, iLQR and DDP use local linearization or second-order expansion of the dynamics together with a backward Riccati recursion, producing both feedforward and feedback terms. Therefore, iLQR is fundamentally a cost-descent method: it seeks a locally optimal trajectory by iteratively improving a quadratic approximation of the value function along the horizon. MPSP, in contrast, does not propagate costates and does not solve a backward dynamic-programming recursion. Its optimization structure is simpler because only the terminal sensitivity map appears in the subproblem, and the terminal condition is imposed as a hard constraint. This makes MPSP cheaper per iteration and especially attractive for short-horizon aggressive maneuvers where terminal accuracy is the main requirement. In a numerical comparison, one would therefore expect MPSP to reach terminal feasibility faster and with a simpler update structure, whereas iLQR would typically provide a more standard local optimal-control framework with explicit feedback gains but at higher computational cost due to the repeated backward Riccati sweep.
\begin{figure*}[!tbp]
  \centering
  \begin{subfigure}{0.48\textwidth}
      \centering
      \includegraphics[width=\textwidth]{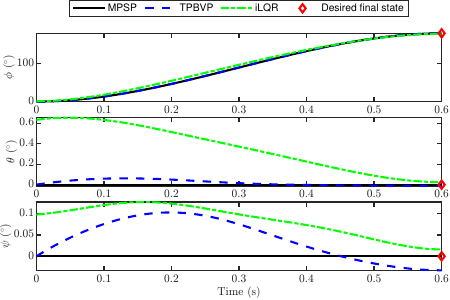}
      \caption{VPQ Euler angle trajectories}
      \label{fig:Euler_ang_VPQ}
  \end{subfigure}
  \hfill
  \begin{subfigure}{0.48\textwidth}
      \centering
      \includegraphics[width=\textwidth]{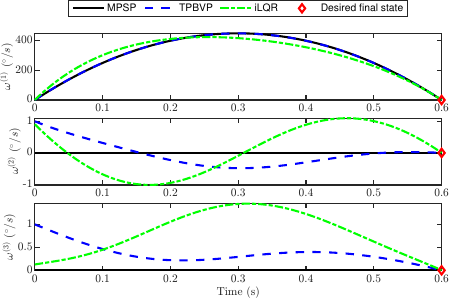}
      \caption{VPQ angular velocity trajectories}
      \label{fig:Ang_velocity_VPQ}
  \end{subfigure}
  \caption{VPQ attitude and angular velocity responses for the $180^\circ$ flip maneuver.}
  \label{fig:VPQ_attitude_two_panel}
\end{figure*}

\begin{figure*}[!tbp]
  \centering
  \begin{subfigure}{0.48\textwidth}
      \centering
      \includegraphics[width=\textwidth]{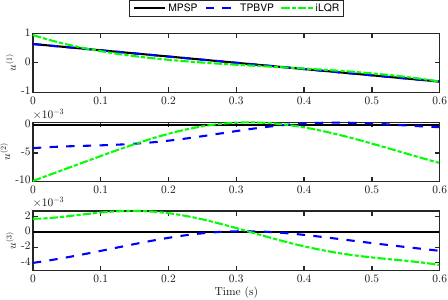}
      \caption{VPQ control inputs}
      \label{fig:Control_input_VPQ}
  \end{subfigure}
  \hfill
  \begin{subfigure}{0.48\textwidth}
      \centering
      \includegraphics[width=\textwidth]{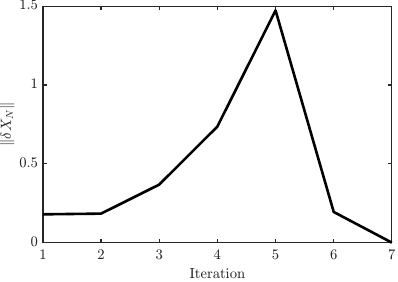}
      \caption{VPQ terminal deviation norm vs iteration}
      \label{fig:VPQ_ErrNorm_itr}
  \end{subfigure}
  \caption{VPQ control inputs and terminal deviation convergence behavior.}
  \label{fig:VPQ_control_two_panel}
\end{figure*}

\subsection{Single Main Rotor Helicopter (SMRH)}

For the SMRH simulations, the discrete time step is $h=0.001$~s and the maneuver horizon is $t_f=1.2$~s, resulting in $N=1200$ steps. The SMRH is commanded to perform the same $180^{\circ}$ roll flip. The initial and final rotor moments are $[l(0),m(0),n(0)]^\top=[0,0,0]^\top$, $[l(t_f),m(t_f),n(t_f)]^\top=[0,0,0]^\top$. Figures~\ref{fig:SRH_Euler_ang}--\ref{fig:SRH_Control_input} show the SMRH results.   The SMRH dynamics include rotor--fuselage coupling, making the continuous-time TPBVP significantly more challenging to solve. Despite this, the MPSP algorithm achieves terminal feasibility within a small number of iterations. As in the VPQ case, the terminal deviation norm may increase in early iterations before entering the local contraction region. The MPSP-generated trajectory aligns closely with the TPBVP benchmark, and the control inputs remain smooth despite the coupled dynamics, demonstrating the capability of the Lie-group MPSP framework to handle more complex actuator–body interactions.


\begin{figure*}[!tbp]
  \centering
  \begin{subfigure}{0.48\textwidth}
      \centering
      \includegraphics[width=\textwidth]{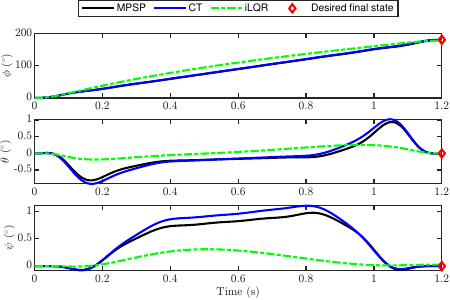}
      \caption{SMRH Euler angle trajectories}
      \label{fig:SRH_Euler_ang}
  \end{subfigure}
  \hfill
  \begin{subfigure}{0.48\textwidth}
      \centering
      \includegraphics[width=\textwidth]{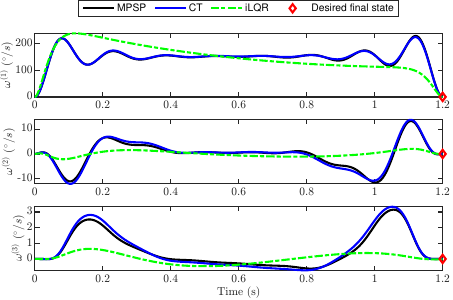}
      \caption{SMRH angular velocity trajectories}
      \label{fig:SRH_Ang_vel}
  \end{subfigure}
  \caption{SMRH attitude and angular velocity responses for the $180^\circ$ flip maneuver.}
  \label{fig:SMRH_attitude_two_panel}
\end{figure*}

\begin{figure*}[!tbp]
  \centering
  \begin{subfigure}{0.48\textwidth}
      \centering
      \includegraphics[width=\textwidth]{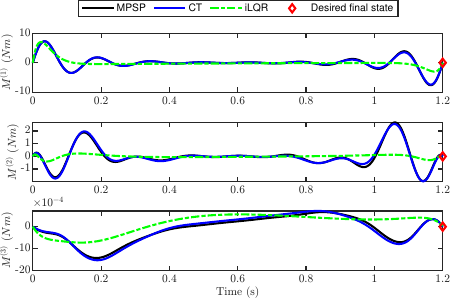}
      \caption{SMRH control inputs}
      \label{fig:SRH_Control_input}
  \end{subfigure}
  \hfill
  \begin{subfigure}{0.48\textwidth}
      \centering
      \includegraphics[width=\textwidth]{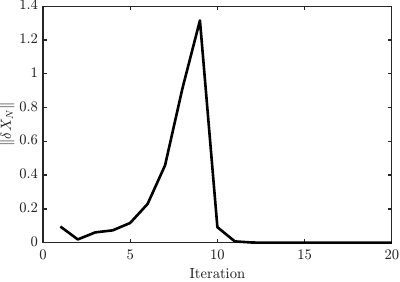}
      \caption{SMRH terminal deviation norm vs iteration}
      \label{fig:SRH_ErrNorm_vs_iter}
  \end{subfigure}
  \caption{SMRH control inputs and terminal deviation convergence behavior.}
  \label{fig:SMRH_control_two_panel}
\end{figure*}

\subsection{Monte Carlo Robustness (Random Initial Control Guess)}

To assess robustness with respect to initialization, Monte Carlo (MC) trials were performed for both VPQ and SMRH. In each trial, the initial state was held fixed and the initial control sequence was randomized (bounded perturbations) to generate diverse initial guesses. Each trial was then solved using the same MPSP settings as the nominal case. Figures~\ref{fig:VPQ_MC_TermErr}--\ref{fig:VPQ_MC_JuBox} summarize the VPQ MC results. All trials converge and exhibit the characteristic feasibility-driven behavior: the terminal deviation may temporarily increase, while the control-increment cost $J_{\Delta u}$ peaks in intermediate iterations and collapses to zero once feasibility is achieved. The final control effort $J_u$ shows very low dispersion across trials, indicating consistent solutions. Figures~\ref{fig:SRH_MC_TermErr}--\ref{fig:SRH_MC_JuBox} show the SMRH MC results. Again, all trials converge successfully, and the dispersion in final $J_u$ remains small, demonstrating that the feasibility-driven MPSP updates lead to repeatable solutions despite randomized initial control guesses. Overall, the MC results confirm robust convergence of the Lie-group MPSP algorithm, numerical stability of the intrinsic sensitivity model, and repeatability of the final control solution across diverse initializations.
\subsection{ Monte Carlo Robustness (Random Initial Conditions)}

To assess robustness with respect to initial condition uncertainty, we perform a Monte Carlo study by randomizing the initial states of both VPQ and SMRH and evaluating the convergence behavior of iLQR and MPSP across multiple trials. For the VPQ, the initial Euler angles are perturbed uniformly within $\pm 10^\circ$ and the initial angular velocities within $\pm 10^\circ/\mathrm{s}$. Under these perturbations, the iLQR method typically drives the terminal deviation norm $\Vert \partial  X_N \Vert$ close to zero within approximately five iterations for many realizations. However, in a non-negligible subset of trials, we observe a steady-state terminal error of about $0.5$, reflecting the fact that $\Vert \partial X_N \Vert = 0$ is not enforced as a hard constraint in the iLQR formulation but only penalized via the terminal cost. In addition, for some initial conditions the terminal error initially increases over a few iterations before eventually decreasing and reaching small values after roughly 10 iterations, indicating sensitivity to the initial conditions and the controller gains. In contrast, for MPSP applied to the same VPQ Monte Carlo samples, the terminal deviation $\Vert \partial X_N \Vert$ converges to machine-zero within at most 10 iterations for all tested initial conditions, and the evolution of $\Vert \partial X_N \Vert$ across iterations exhibits a consistent, nearly monotone pattern. This uniform behavior is a direct consequence of enforcing the terminal condition as a hard equality constraint in each MPSP subproblem, and it translates into more reliable and predictable performance over random initializations.

For the SMRH, we adopt the same randomization range for the Euler angles and angular velocities as in the VPQ case, and additionally perturb the initial rotor/fuselage moment $M(0)$ within $\pm 0.1\,\mathrm{Nm}$. Under these conditions, the iLQR algorithm reduces $\Vert \partial X_N \Vert$ to near zero in roughly five iterations for all tested realizations, while MPSP typically requires about fifteen iterations to achieve the same terminal accuracy. In this setting, both methods ultimately converge, but their qualitative behavior differs: MPSP exhibits similar, repeatable convergence patterns across trials, whereas the iLQR convergence curve is more sensitive to the choice of gains (e.g., terminal and running weights, line-search parameters) and can show more variability from one randomized initial condition to another. Overall, the Monte Carlo results suggest that, while iLQR can converge faster in terms of iteration count when appropriately tuned, MPSP offers more consistent and robust convergence properties under random initial conditions, especially when strict enforcement of the terminal constraint is critical.

\begin{figure*}[!tbp]
  \centering
  \begin{subfigure}{0.48\textwidth}
      \centering
      \includegraphics[width=\textwidth]{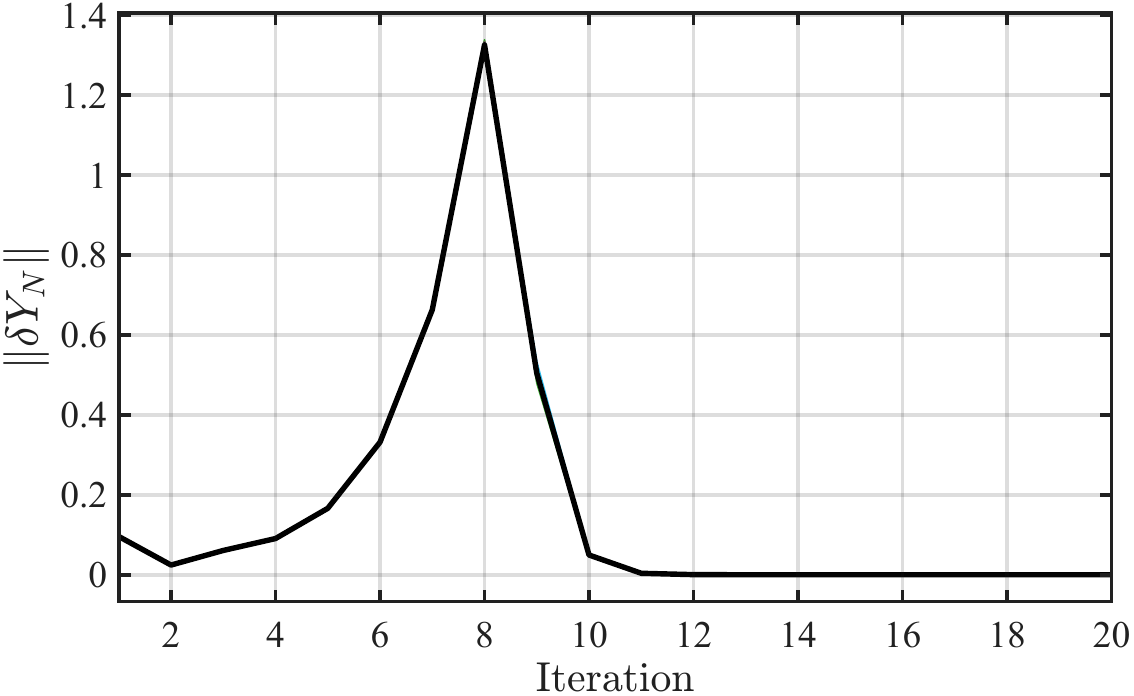}
      \caption{VPQ MC: terminal deviation norm $\|\delta X_N\|$ vs iteration}
      \label{fig:VPQ_MC_TermErr}
  \end{subfigure}
  \hfill
  \begin{subfigure}{0.48\textwidth}
      \centering
      \includegraphics[width=\textwidth]{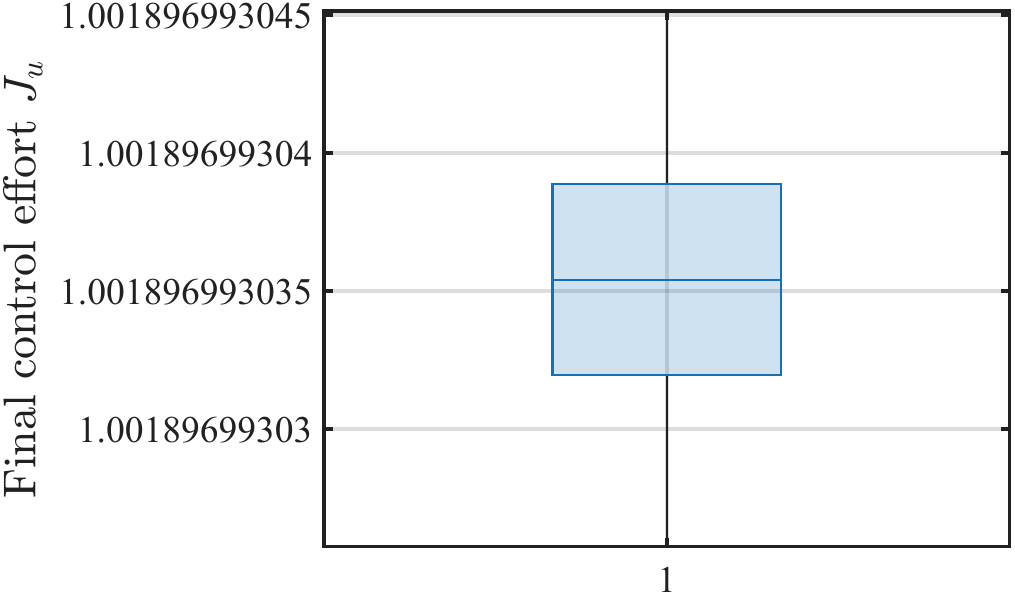}
      \caption{VPQ MC: control increment cost $J_{\Delta u}$ vs iteration}
      \label{fig:VPQ_MC_JdU}
  \end{subfigure}
  \caption{VPQ Monte Carlo results: terminal deviation and control-increment cost across randomized initial guesses.}
  \label{fig:VPQ_MC_two_panel}
\end{figure*}

\begin{figure*}[!tbp]
  \centering
  \begin{subfigure}{0.55\textwidth}
      \centering
      \includegraphics[width=\textwidth]{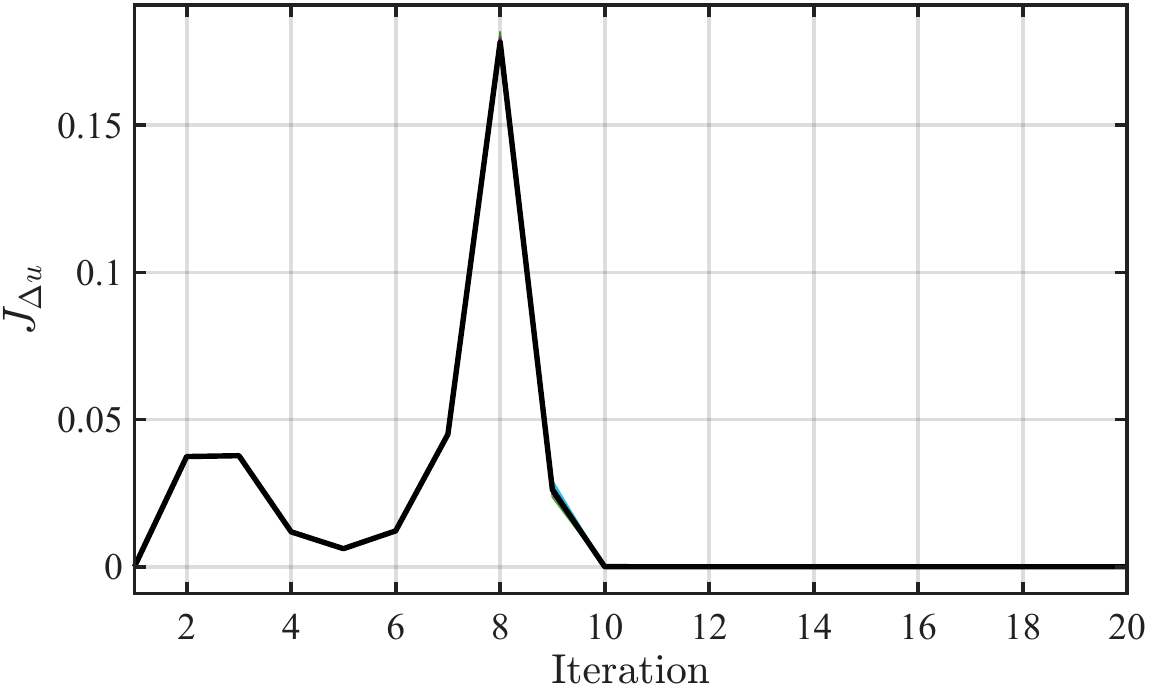}
      \caption{VPQ MC: boxplot of final control effort $J_u$}
      \label{fig:VPQ_MC_JuBox}
  \end{subfigure}
  \caption{VPQ Monte Carlo: distribution of final control effort across randomized initial guesses.}
  \label{fig:VPQ_MC_box_only}
\end{figure*}


\begin{figure*}[!tbp]
  \centering
  \begin{subfigure}{0.48\textwidth}
      \centering
      \includegraphics[width=\textwidth]{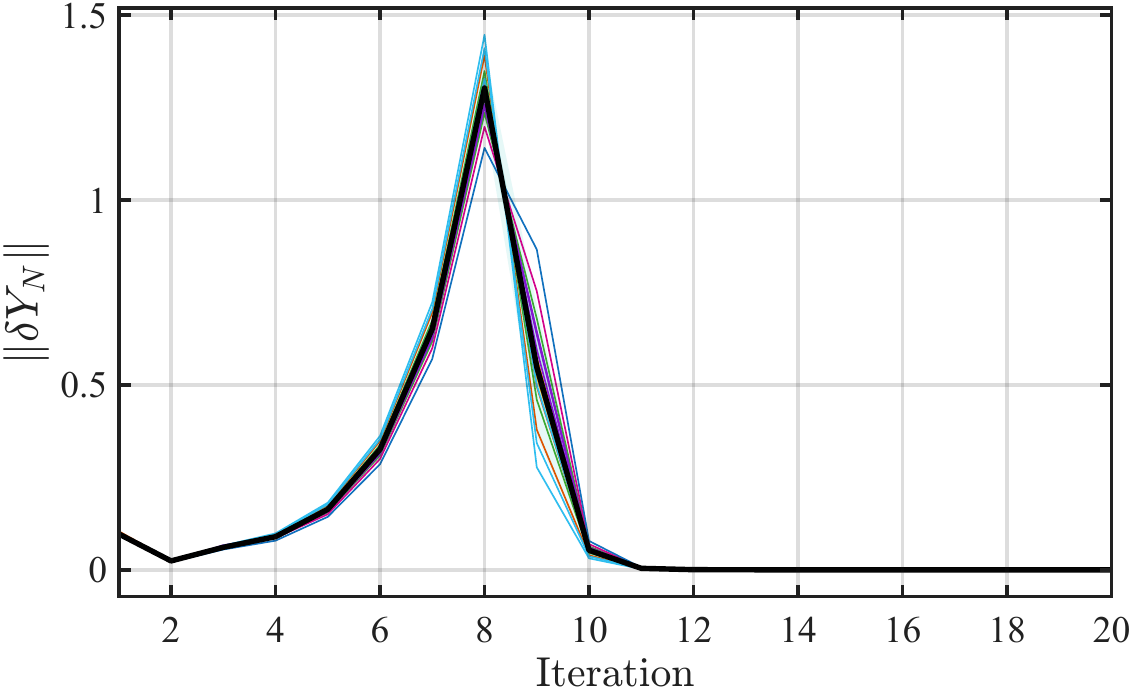}
      \caption{SMRH MC: terminal deviation norm $\|\delta X_N\|$ vs iteration}
      \label{fig:SRH_MC_TermErr}
  \end{subfigure}
  \hfill
  \begin{subfigure}{0.48\textwidth}
      \centering
      \includegraphics[width=\textwidth]{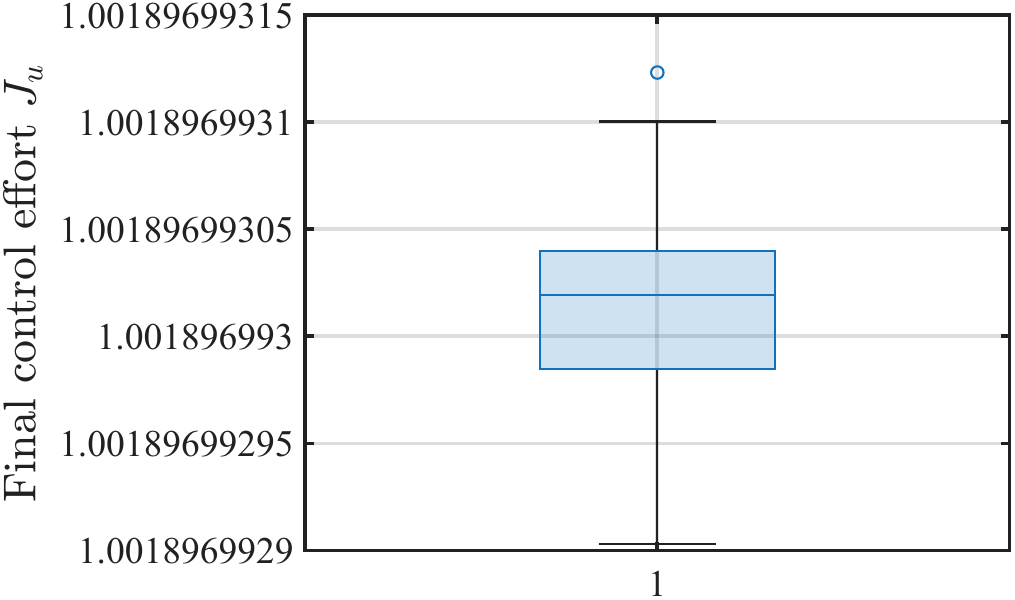}
      \caption{SMRH MC: control increment cost $J_{\Delta u}$ vs iteration}
      \label{fig:SRH_MC_JdU}
  \end{subfigure}
  \caption{SMRH Monte Carlo results: terminal deviation and control-increment cost across randomized initial guesses.}
  \label{fig:SRH_MC_two_panel}
\end{figure*}

\begin{figure*}[!tbp]
  \centering
  \begin{subfigure}{0.55\textwidth}
      \centering
      \includegraphics[width=\textwidth]{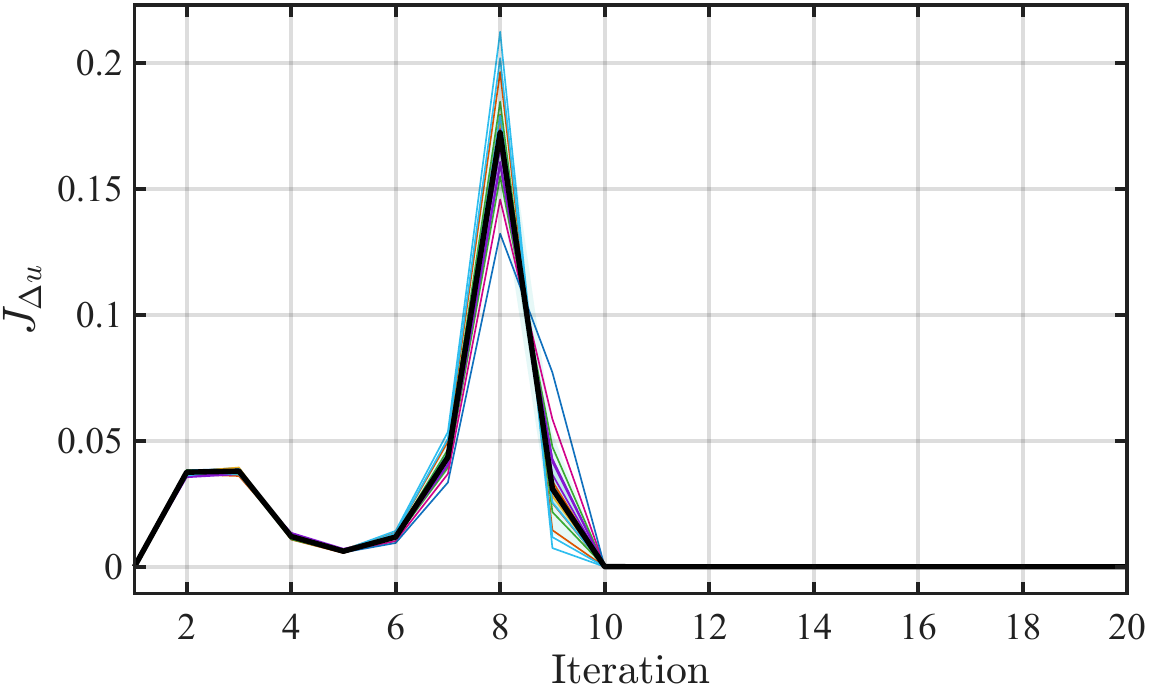}
      \caption{SMRH MC: boxplot of final control effort $J_u$}
      \label{fig:SRH_MC_JuBox}
  \end{subfigure}
  \caption{SMRH Monte Carlo: distribution of final control effort across randomized initial guesses.}
  \label{fig:SRH_MC_box_only}
\end{figure*}
\begin{figure*}[!tbp]
  \centering
  \begin{subfigure}{0.48\textwidth}
      \centering
      \includegraphics[width=\textwidth]{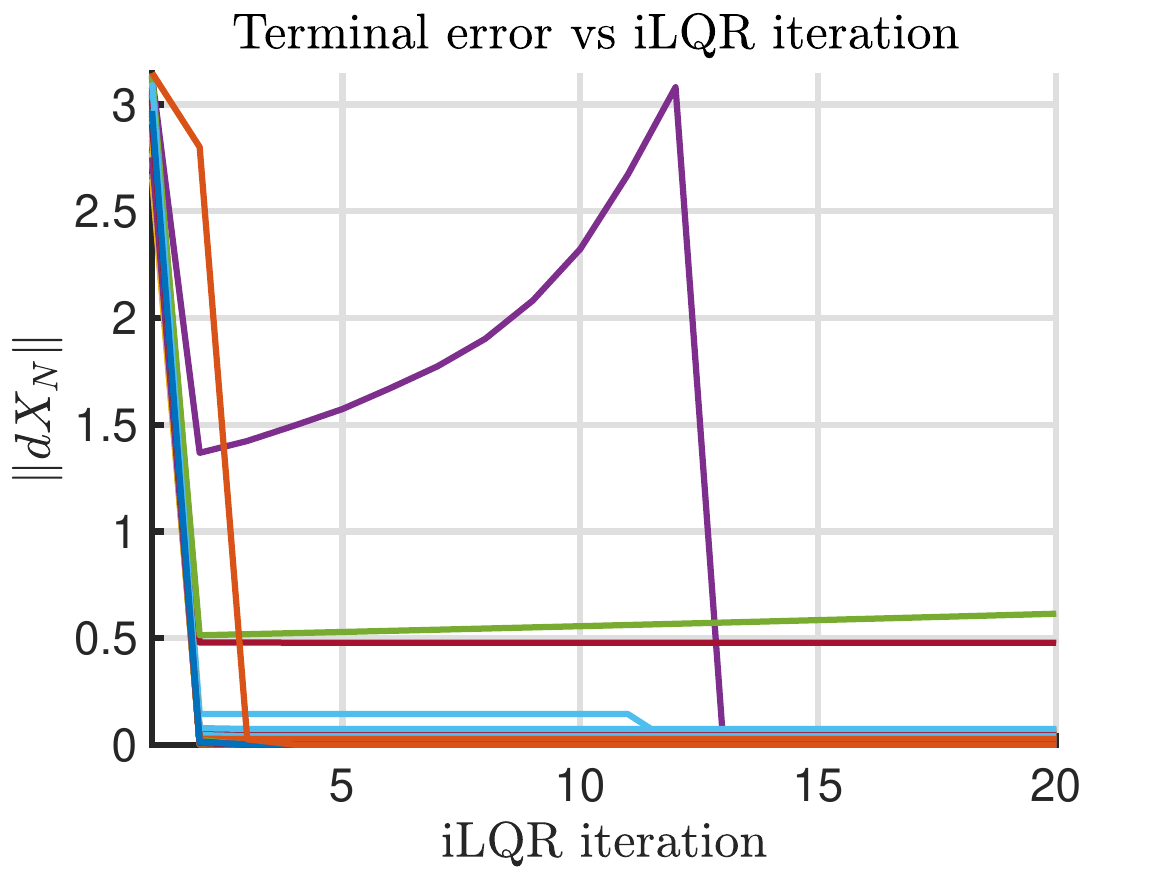}
      \caption{VPQ MC iLQR: terminal deviation norm $\|\delta X_N\|$ vs iteration}
      \label{fig:VPQ_MC_iLQR_IC}
  \end{subfigure}
  \hfill
  \begin{subfigure}{0.48\textwidth}
      \centering
      \includegraphics[width=\textwidth]{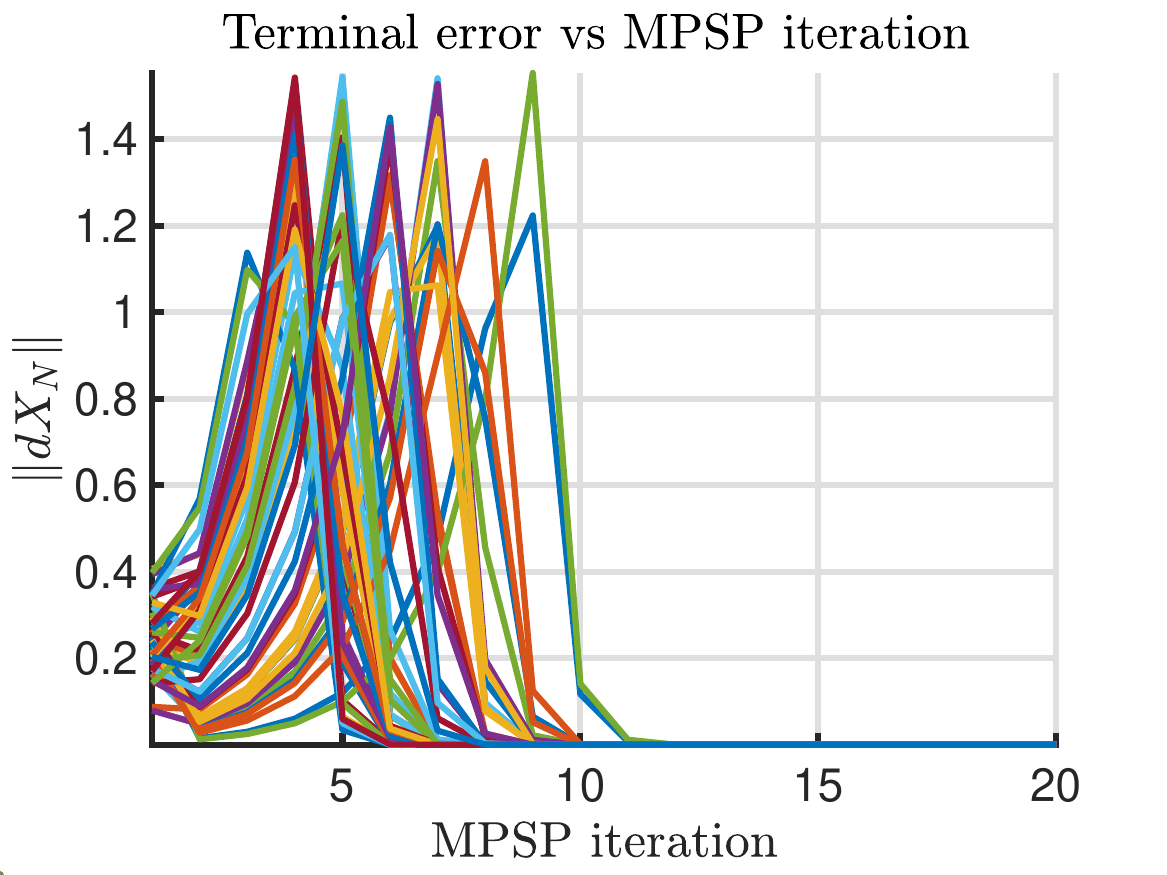}
      \caption{VPQ MC MPSP: terminal deviation norm $\|\delta X_N\|$ vs iteration}
      \label{fig:VPQ_MC_MPSP_IC}
  \end{subfigure}
\end{figure*}
\begin{figure*}[!tbp]
  \centering
  \begin{subfigure}{0.48\textwidth}
      \centering
      \includegraphics[width=\textwidth]{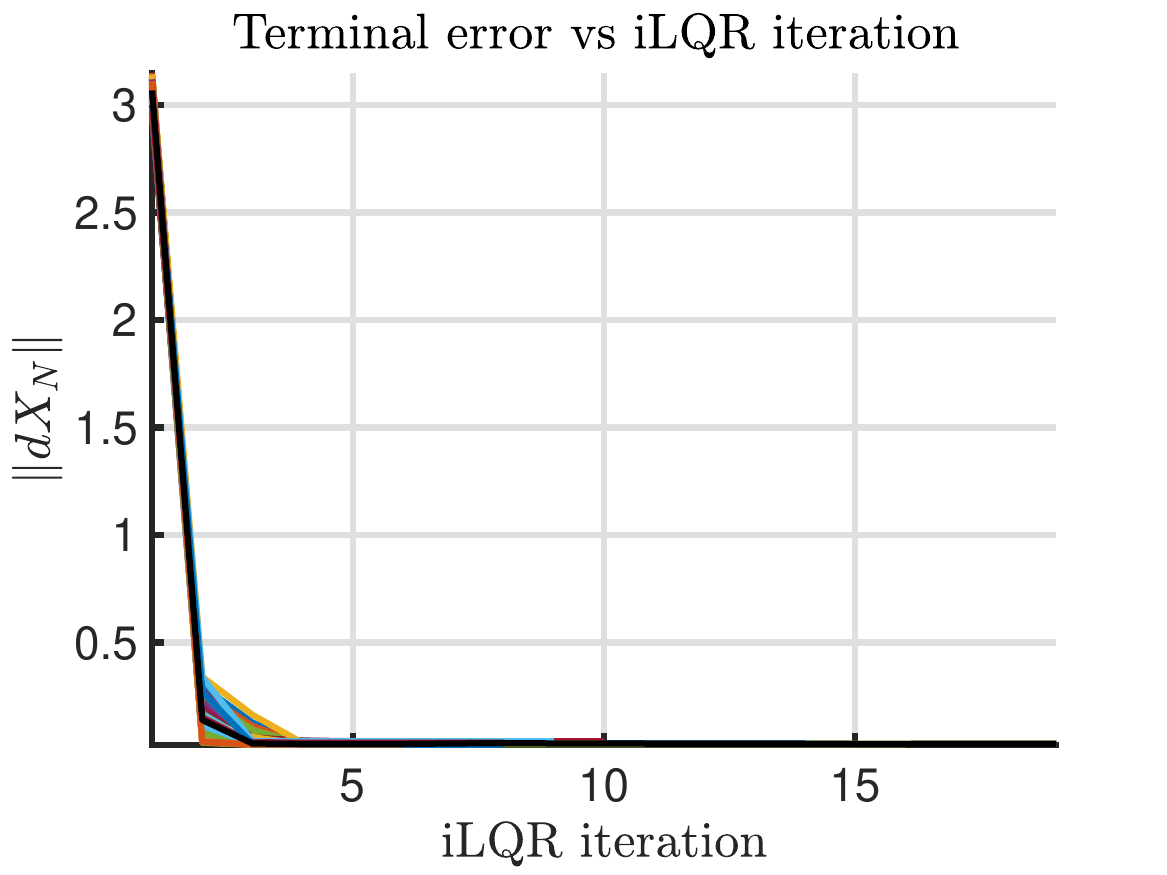}
      \caption{SMRH MC iLQR: terminal deviation norm $\|\delta X_N\|$ vs iteration}
      \label{fig:SMRH_MC_iLQR_IC}
  \end{subfigure}
  \hfill
  \begin{subfigure}{0.48\textwidth}
      \centering
      \includegraphics[width=\textwidth]{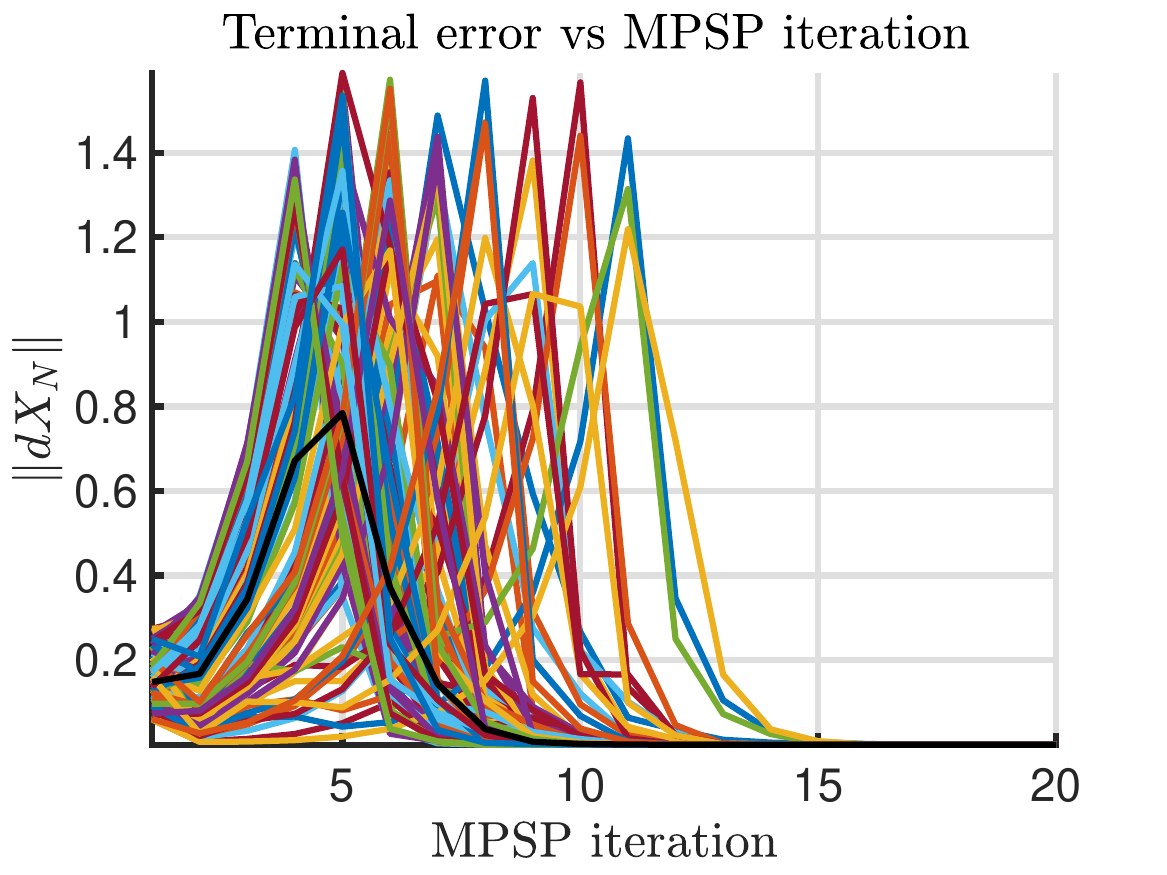}
      \caption{SMRH MC MPSP: terminal deviation norm $\|\delta X_N\|$ vs iteration}
      \label{fig:SMRH_MC_MPSP_IC}
  \end{subfigure}
\end{figure*}
It is worth noting that in several earlier works, MPSP has been successfully employed as an efficient initializer for solving nonlinear two‑point boundary value problems arising from the Pontryagin Maximum Principle. Because MPSP provides a dynamically feasible control sequence with small terminal error at very low computational cost, it often serves as a high‑quality initial guess that significantly improves the convergence of shooting‑based TPBVP solvers. In the present work, although our primary goal is to demonstrate the standalone performance of the Lie‑group MPSP algorithm, the numerical results also indicate that the generated trajectories are sufficiently close to the continuous‑time optimal solution to serve as reliable initializations for PMP‑based methods.


\section{Conclusion}\label{Sec:Conclusion}
In this paper, an extension of the computationally efficient Model Predictive Static Programming (MPSP) algorithm for simple mechanical systems evolving on Lie groups has been presented. The algorithm provides a closed-form update for the optimal control input by converting the finite-time optimal control problem into a static optimization problem, leading to a static Lagrange multiplier and avoiding the need to solve any two-point boundary value problem. The proposed formulation leverages intrinsic linearization, left-trivialized variations, and recursive computation of terminal sensitivity matrices to enable real-time implementation on nonlinear configuration spaces such as Lie groups. By avoiding the numerical difficulties associated with shooting-based TPBVP solvers, the method offers a practical alternative for aggressive finite-time maneuvers. The algorithm was numerically validated through optimal flipping maneuvers of a variable-pitch quadrotor (VPQ) and a single-main-rotor helicopter (SMRH), with the resulting state and control trajectories compared against those obtained from the continuous-time sufficiency conditions and iLQR technique. The simulation results demonstrate that the MPSP-generated trajectories closely match the continuous-time optimal solutions, with negligible deviation in both attitude and control histories. This agreement confirms that the intrinsic deviation model accurately captures the local behavior of the dynamics on Lie groups and that the MPSP updates remain stable even for highly nonlinear, fast maneuvers. Furthermore, the results reinforce a trend observed in earlier literature: MPSP not only serves as a standalone real-time controller but also provides high-quality initial guesses for solving the nonlinear TPBVPs arising from the Pontryagin Maximum Principle. The near-coincidence of MPSP and continuous-time solutions in our simulations suggests that the proposed Lie-group extension can similarly be used to initialize PMP-based solvers for improved convergence and robustness.  Overall, the study establishes MPSP as an effective, computationally efficient, and geometrically consistent framework for finite-time optimal control of mechanical systems evolving on Lie groups.

\bibliography{main.bbl}

\end{document}